\documentclass[jair,twoside,11pt,theapa]{article}
\usepackage{jair, rawfonts}

\firstpageno{1}

\usepackage{times}

\usepackage{amsmath}
\usepackage{amsfonts}
\usepackage{amssymb}
\usepackage{amsthm}
\usepackage{amsbsy}
\usepackage{graphicx}
\usepackage{subcaption}
\usepackage{verbatim}
\usepackage{url}
\usepackage{bbm}
\usepackage{multirow}
\usepackage{bm}

\usepackage{tikz}
\usetikzlibrary{shapes,snakes}
\usetikzlibrary{decorations}
\usepackage{pgfplots}
\pgfplotsset{compat=1.10}
\usepgfplotslibrary{fillbetween}

\usepackage[ruled]{algorithm2e}

\SetArgSty{textrm}  

\usepackage{ctable}
\usepackage{natbib}

\newtheorem{theorem}{Theorem}
\newtheorem{lemma}[theorem]{Lemma}
\newtheorem{claim}[theorem]{Claim}

\theoremstyle{definition}
\newtheorem{example}{Example}

\newcommand{\remove}[1]{}

\newcommand{\R}{\mathbb{R}}

\newcommand{\EE}{\mathbb{E}}

\newcommand{\bigO}{\mathcal{O}}

\newcommand{\one}[1]{\mathbbm{1}\{#1\}}

\newcommand{\union}{\cup}

\newcommand{\ttt}{\mathbf{t}}
\newcommand{\sss}{\mathbf{s}}
\newcommand{\zz}{\mathbf{z}}
\DeclareMathOperator{\cost}{cost}
\DeclareMathOperator{\weight}{weight}
\newcommand{\PNE}{\text{PNE}}
\newcommand{\MNE}{\text{MNE}}
\DeclareMathOperator{\SC}{\text{SC}}

\newcommand{\PoA}{\text{PoA}}
\newcommand{\PoS}{\text{PoS}}
\newcommand{\mPoA}{\text{MPoA}}
\newcommand{\mPoS}{\text{MPoS}}

\DeclareMathOperator{\excess}{excess}

\begin{document}

\title{Bounding the inefficiency of compromise  in opinion formation\thanks{A preliminary version of this paper entitled ``Bounding the inefficiency of compromise'' appeared in {\em Proceedings of the 26th International Joint Conference on Artificial Intelligence (IJCAI)}, pages 142-148, 2017. This work was partially supported by COST Action CA 16228 ``European Network for Game Theory'' and by a PhD scholarship from the Onassis Foundation.}}

\author{\name Ioannis Caragiannis \email \normalfont caragian@ceid.upatras.gr \\
       \addr University of Patras
       \AND
       \name Panagiotis Kanellopoulos \email \normalfont kanellop@ceid.upatras.gr \\
       \addr University of Patras \& CTI ``Diophantus''
       \AND
       \name Alexandros A. Voudouris \email \normalfont voudouris@ceid.upatras.gr \\
       \addr University of Patras}


\maketitle

\begin{abstract}
Social networks on the Internet have seen an enormous growth recently and play a crucial role in different aspects of today's life. They have facilitated information dissemination in ways that have been beneficial for their users but they are often used strategically in order to spread information that only serves the objectives of particular users. These properties have inspired a revision of classical opinion formation models from sociology using game-theoretic notions and tools. We follow the same modeling approach, focusing on scenarios where the opinion expressed by each user is a compromise between her internal belief and the opinions of a small number of neighbors among her social acquaintances. We formulate simple games that capture this behavior and quantify the inefficiency of equilibria using the well-known notion of the price of anarchy. Our results indicate that compromise comes at a cost that strongly depends on the neighborhood size.
\end{abstract}

\section{Introduction}\label{sec:intro}
{\em Opinion formation} has been the subject of much research in sociology, economics, physics, and epidemiology. Due to the widespread adoption of the Internet and the subsequent blossoming of social networks, it has recently attracted the interest of researchers in AI (e.g., see \cite{ACF+16,SIB+15,TL14}) and CS at large (e.g., see \cite{BKO15,MT14,OT09}) as well.

An influential model that captures the adoption of opinions in a social context has been proposed by \cite{FJ90}. According to this, each individual has an internal belief on an issue and publicly expresses a (possibly different) opinion; internal beliefs and public opinions are modeled as real numbers. In particular, the opinion that an individual expresses follows by averaging between her internal belief and the opinions expressed by her social acquaintances. Recently, \cite{BKO15} show that this behavior can be explained through a game-theoretic lens: averaging between the internal belief of an individual and the opinions in her social circle is simply a {\em strategy} that minimizes an implicit cost for the individual.

\citet{BKO15} use a quadratic function to define this cost. Specifically, this function is equal to the total squared distance of the opinion that the individual expresses from her belief and the opinions expressed in her social circle. In a sense, this behavior leads to opinions that follow the majority of her social acquaintances. \cite{BKO15} consider a static snapshot of the social network. In contrast, \citet{BGM13} implicitly assume that the opinion of an individual depends on a small number of acquaintances only, her {\em neighbors}. So, in their model, opinion formation {\em co-evolves} with the neighborhood for each individual: her neighborhood consists of those who have opinions similar to her belief. Then, the opinion expressed is assumed to minimize the same cost function used by \citet{BKO15}, taking into account the neighborhood instead of the whole social circle.

We follow the co-evolutionary model of \citet{BGM13}, but we deviate from their cost definition and instead consider individuals that seek to {\em compromise} with their neighbors. Hence, we assume that each individual aims to minimize the {\em maximum} distance of her expressed opinion from her belief and each of her neighbors' opinion. Like \citet{BGM13}, we assume that opinion formation co-evolves with the social network. Each individual's neighborhood consists of the $k$ other individuals with the closest opinions to her belief. Naturally, these modeling decisions lead to the definition of strategic games, which we call $k$-{\em compromising opinion formation} (or, simply, $k$-COF) games. Each individual is a (cost-minimizing) player with the opinion expressed as her strategy.

\subsection{Technical contribution}
We study questions related to the existence, computational complexity, and quality of equilibria in $k$-COF games. We begin by proving several properties about the geometric structure of opinions and beliefs at pure Nash equilibria, i.e., in states of the game where each player minimizes her individual cost assuming that the remaining players will not change their opinions, and, thus, has no incentive to deviate by expressing a different opinion.

Using these structural properties we show that there exist simple $k$-COF games that do not admit pure Nash equilibria. Furthermore, we prove that even in games where equilibria do exist, their quality may be suboptimal in terms of the {\em social cost}, i.e., the total cost experienced by all players. To quantify this inefficiency, we show that the optimistic measure known as the {\em price of stability} (introduced by \citet{adktwr09SIAM}) which, informally, is defined as the ratio of the minimum social cost achieved at any pure Nash equilibrium to the minimum social cost at any possible state of the game, grows linearly with $k$.

For the special case of $1$-COF games, we show that each such game admits a representation as a directed acyclic graph, in which every pure Nash equilibrium corresponds to a path between two designated nodes. Hence, the problems of computing the best or worst (in terms of the social cost) pure Nash equilibrium (or even of computing whether such an equilibrium exists) are equivalent to simple path computations that can be performed in polynomial time.

For general $k$-COF games, we quantify the inefficiency of the worst-case pure Nash equilibria by bounding the pessimistic measure known as the {\em price of anarchy} (introduced by \citet{KP99}) which, informally, is defined as the ratio of the maximum social cost achieved at any pure Nash equilibrium to  the minimum possible social cost at any state of the game. Specifically, we present upper and lower bounds on the price of anarchy of $k$-COF games (with respect to both pure and mixed Nash equilibria) that suggest a linear dependence on $k$. Our upper bound on the price of anarchy exploits, in a non-trivial way, linear programming duality in order to lower-bound the optimal social cost. For the fundamental case of $1$-COF games, we obtain a tight bound of $3$ using a particular charging scheme in the analysis. Our contribution is summarized in Table~\ref{tab:poa-results}.

\begin{table}[h!]
\centering
\small \begin{tabular}{c c c c  p{6cm}}
\noalign{\hrule height 1pt}\hline
$k$ 	& \PoA & \mPoA & \PoS 		& Existence/Complexity \\
\noalign{\hrule height 1pt}\hline
$1$		&	\multicolumn{1}{l}{$3$ (Thms \ref{thm:poa-upper-1}, \ref{thm:pure-lower-1})} & $\geq 6$ & $\geq 17/15$ & PNE may not exist for any $k$ (Thm. \ref{thm:no-pure})\\ 	
& & (Thm. \ref{thm:mixed-lower-1})& (Thm. \ref{thm:pos}) &Best and worst PNE is in {\cal P} (Thm. \ref{thm:complexity})\\
& & & & \\
$2$		&	\multicolumn{1}{l}{$\leq 12$ (Thm. \ref{thm:ppoa-upper})}	&	$\geq 24/5$	& $\geq 8/7$	& Open question: Is computing a PNE\\ 	
& \multicolumn{1}{l}{$\geq 18/5$ (Thm. \ref{thm:pure-lower-k})} & (Thm. \ref{thm:mixed-lower-k})& (Thm. \ref{thm:pos-2}) & in {\cal P} when $k\geq 2$?\\
& & & & \\
$\geq 3$	&	\multicolumn{1}{l}{$\leq 4(k+1)$ (Thm. \ref{thm:ppoa-upper})}	&	$\geq k+2$	& $\geq \frac{k+1}{3}$ &  \\ 				
& \multicolumn{1}{l}{$\geq k+1$ (Thm. \ref{thm:pure-lower-k})}& (Thm. \ref{thm:mixed-lower-k})& (Thm. \ref{thm:pos-k}) &\\
\noalign{\hrule height 1pt}\hline
\end{tabular}
\caption{Summary of our results for $k$-COF games. The table presents our bounds on the price of anarchy over pure Nash equilibria (\PoA) and mixed Nash equilibria (\mPoA), on the price of stability (\PoS) as well as the existence and complexity of pure Nash equilibria (PNE). Clearly, any upper bound on \PoA{} is also an upper bound on \PoS.}
\label{tab:poa-results}
\end{table}

\subsection{Related work}
\citet{D74} proposed a framework that models the opinion formation process, where each individual updates her opinion based on a weighted averaging procedure. Subsequently, \citet{FJ90} refined the model by assuming that each individual has a private belief and expresses a (possibly different) public opinion that depends on her belief and the opinions of people to whom she has social ties. More recently, \citet{BKO15} studied this model and proved that, for the setting where beliefs and opinions are in $[0,1]$, the repeated averaging process leads to an opinion vector that can be thought of as the unique equilibrium in a corresponding opinion formation game.

Deviating from the assumption that opinions depend on the whole social circle, \citet{BGM13} consider co-evolutionary opinion formation games, where as opinions evolve so does the neighborhood of each person. This model is conceptually similar to previous ones that have been studied by \citet{HK02}, and \citet{HN06}. Both \citet{BKO15} and \citet{BGM13} show constant bounds on the price of anarchy of the games that they study. In contrast, the modified cost function we use in order to model compromise yields considerably higher price of anarchy.

A series of recent papers from the EconCS community consider discrete models with binary opinions.  \citet{CKO18} consider discrete preference games, where beliefs and opinions are binary and study questions related to the price of stability. For these games, \citet{ACFGP15,ACFGP17} characterize the social networks where the belief of the minority can emerge as the opinion of the majority, while \cite{ACFG17} examine the robustness of such results to variants of the model. \citet{ACF+16} generalize discrete preference games so that players are not only interested in agreeing with their neighbors and more complex constraints can be used to represent the players' preferences. \citet{BFM16} extend co-evolutionary formation games to the discrete setting. Other models assume that opinion updates depend on the entire social circle of each individual, who consults a small random subset of social acquaintances; see the recent paper by \citet{FPS16} and the survey of \citet{MT14}.

When there are more than one issues to be discussed, \cite{JMFB15} propose and analyze the DeGroot-Friedkin model for the evolution of an influence network between individuals who form opinions on a sequence of issues, while \cite{XLB15} introduce a modification to the DeGroot-Friedkin model so that each individual may recalculate the weight given to her opinion, i.e., her self-confidence, after the discussion of each issue.

Another line of research considers how fast a system converges to a stable state. In this spirit, \cite{EB15} consider the dynamics of the Hegselmann-Krause model \citeyearpar{HK02}, where opinions and neighborhoods co-evolve, and study the termination time in finite dimensions under different settings. Similarly, \cite{FGV16} study the speed of convergence of decentralized dynamics in finite opinion games, where players have only a finite number of opinions available. \cite{FV17} consider the role of social pressure towards consensus in opinion games and provide tight bounds on the speed of convergence for the important special case where the social network is a clique.

\cite{DGM14} perform a set of online user studies and argue that widely studied theoretical models do not completely explain the experimental results obtained. Hence, they introduce an analytical model for opinion formation and present preliminary theoretical and simulation results on the convergence and structure of opinions when users iteratively update their respective opinions according to the new model.

\cite{Cha12} analyzes influence systems, where each individual observes the location of her neighbors and moves accordingly, and presents an algorithmic calculus for studying such systems.  \cite{KKOS16} present a novel model of cultural dynamics and study the interplay between selection and influence. Their results include an almost complete characterization of stable outcomes and guaranteed convergence from all starting states. \cite{GLK12} consider network diffusion and contagion propagation. Their goal is to infer an unknown network over which contagion propagated, tracing paths of diffusion and influence. Finally, \cite{KKT15} study the optimization problem for influence maximization in a social networks, where each individual may decide to adopt an idea or an innovation depending on how many of her neighbors already do. The goal is to select an initial seed set of early adopters so that the number of adopters is maximized.

In spite of the extensive related literature on opinion formation in many different disciplines, our model introduces a novel cost function that, we believe, captures the tendency of individuals to compromise more accurately.

\subsection{Roadmap}
We begin with preliminary definitions and notation in Section~\ref{sec:preliminaries}. Then, in Section~\ref{sec:properties} we present several structural properties of pure Nash equilibria, while Section~\ref{sec:existence} is devoted to the existence and the price of stability of these equilibria. In Section~\ref{sec:algorithm}, we present an algorithm that determines whether pure Nash equilibria exist in a $1$-COF game, and, in addition, computes the best and worst such equilibria, when they exist. In Sections~\ref{sec:upper} and~\ref{sec:upper-1} we prove upper bounds on the price of anarchy of $k$-COF and $1$-COF games, respectively, while Section~\ref{sec:lower} contains our lower bounds. We conclude in Section~\ref{sec:discussion} with a discussion of open problems and possible extensions of our work.

\section{Definitions and notation}\label{sec:preliminaries}
A compromising opinion formation game defined by the $k$ nearest neighbors (henceforth, called $k$-COF game) is played by a set of $n$ players whose beliefs lie on the line of real numbers. Let $\sss=(s_1,s_2,\dots,s_n) \in \R^n$ be the vector containing the players' beliefs such that $s_i \leq s_{i+1}$ for each $i \in [n-1]$. Let $\zz=(z_1,z_2,\dots,z_n) \in \R^n$ be a vector containing the (deterministic or randomized) opinions expressed by the players; these opinions define a state of the game. We denote by $\zz_{-i}$ the opinion vector obtained by removing $z_i$ from $\zz$. In an attempt to simplify notation, we omit $k$ from all relevant definitions.

Given vector $\zz$ (or a realization of it in case $\zz$ contains randomized opinions), we define the neighborhood $N_i(\zz, \sss)$ of player $i$ to be the set of $k$ players whose opinions are the closest to the belief of player $i$ breaking ties arbitrarily (but consistently). For each player $i$, we define $I_i(\zz, \sss)$ as the shortest interval of the real line that includes the following points: the belief $s_i$, the opinion $z_i$, and the opinion $z_j$ for each player $j \in N_i(\zz, \sss)$. Furthermore, let $\ell_i(\zz, \sss)$ and $r_i(\zz, \sss)$ be the players with the leftmost and rightmost point in $I_i(\zz, \sss)$, respectively. For example, $\ell_i(\zz,\sss)$ can be equal to either player $i$ or some player $j\in N_i(\zz,\sss)$, depending on whether the leftmost point of $I_i(\zz,\sss)$ is $s_i$, $z_i$, or $z_j$. To further simplify notation, we will frequently use $\ell(i)$ and $r(i)$ instead of $\ell_i(\zz,\sss)$ and $r_i(\zz, \sss)$ when $\zz$ and $\sss$ are clear from the context. In the following, we present the relevant definitions for the case of possibly randomized opinion vectors; clearly, these can be simplified whenever $\zz$ consists entirely of deterministic opinions.

Given a $k$-COF game with belief vector $\sss$, the cost that player $i$ experiences at the state of the game defined by an opinion vector $\zz$ is
\begin{align}\label{eq:costdefinition}\nonumber
\EE[\cost_i(\zz, \sss)] &= \EE\left[\max_{j \in N_i(\zz,\sss)}\bigg\{ |z_i-s_i|, |z_j-z_i| \bigg\}\right] \\
&= \EE\left[\max\bigg\{ |z_i-s_i|, |z_{r_i(\zz,\sss)}-z_i|, |z_i-z_{\ell_i(\zz,\sss)}| \bigg\}\right].
\end{align}
For the special case of $1$-COF games, we denote by $\sigma_i(\zz,\sss)$ (or $\sigma(i)$ when $\zz$ and $\sss$ are clear from the context) the player (other than $i$) whose opinion is closest to the belief $s_i$ of player $i$; notice that $\sigma(i)$ is the only member of $N_i(\zz,\sss)$. In this case, the cost of player $i$ can be simplified as
\begin{align}\label{eq:costdefinition-1}
\EE[\cost_i(\zz, \sss)] = \EE\left[\max\bigg\{ |z_i-s_i|, |z_{\sigma_i(\zz,\sss)}-z_i| \bigg\}\right].
\end{align}

We say that an opinion vector $\zz$ consisting entirely of deterministic opinions is a {\em pure Nash equilibrium} if no player $i$ has an incentive to unilaterally deviate to a deterministic opinion $z_i'$ in order to decrease her cost, i.e.,
$$\cost_i(\zz, \sss) \leq \cost_i((z_i',\zz_{-i}), \sss),$$
where by $(z_i',\zz_{-i})$ we denote the opinion vector in which player $i$ chooses the opinion $z_i'$ and all other players choose the opinions they have according to vector $\zz$. Similarly, a possibly randomized opinion vector $\zz$ is a {\em mixed Nash equilibrium} if for any player $i$ and any deviating deterministic opinion $z_i'$ we have
$$\EE[\cost_i(\zz, \sss)] \leq \EE_{\zz_{-i}}[\cost_i((z_i',\zz_{-i}), \sss)].$$
Let $\PNE(\sss)$ and $\MNE(\sss)$ denote the sets of pure and mixed Nash equilibria, respectively, of the $k$-COF game with belief vector $\sss$.

The {\em social cost} of an opinion vector $\zz$ is the total cost experienced by all players, i.e.,
$$\EE[\SC(\zz, \sss)] = \sum_{i=1}^n{ \EE[\cost_i(\zz, \sss)] }.$$
Let $\zz^*(\sss)$ be a deterministic opinion vector that minimizes the social cost for the given $k$-COF game with belief vector $\sss$; we will refer to it as an {\em optimal} opinion vector for $\sss$.

The {\em price of anarchy} (PoA) over pure Nash equilibria of a particular $k$-COF game with belief vector $\sss$ is defined as the ratio between the social cost of its {\em worst} (in terms of the social cost) pure Nash equilibrium and the optimal social cost, i.e.,
$$\PoA(\sss) = \sup_{\zz \in \PNE(\sss)} \frac{\SC(\zz, \sss)}{\SC(\zz^*(\sss), \sss)}.$$
The {\em price of stability} (PoS) over pure Nash equilibria of the $k$-COF game with belief vector $\sss$ is defined as the ratio between the social cost of the {\em best} pure Nash equilibrium (in terms of social cost) and the optimal social cost, i.e.,
$$\PoS(\sss) = \inf_{\zz \in \PNE(\sss)} \frac{\SC(\zz, \sss)}{\SC(\zz^*(\sss), \sss)}.$$

Similarly, the price of anarchy and the price of stability over mixed Nash equilibria of a $k$-COF game with belief vector $\sss$ are defined as
$$\mPoA(\sss) = \sup_{\zz \in \MNE(\sss)} \frac{\EE[\SC(\zz, \sss)]}{\SC(\zz^*(\sss), \sss)}$$
and
$$\mPoS(\sss) = \inf_{\zz \in \MNE(\sss)} \frac{\EE[\SC(\zz, \sss)]}{\SC(\zz^*(\sss), \sss)},$$
respectively.

Then, the price of anarchy and the price of stability of $k$-COF games, for a fixed $k$, are defined as the supremum of $\PoA(\sss)$ and $\PoS(\sss)$ over all belief vectors $\sss$, respectively.

We conclude this section with an example.

\begin{example}\label{ex:prelim-example}
Consider the $1$-COF game with three players and belief vector $\sss = (-10,2,5)$ which is depicted in Figure~\ref{fig:prelim-example}(a). For simplicity, we will refer to the players as left ($\ell$), middle ($m$), and right ($r$).

Let us examine the opinion vector $\zz = (-10,-5,4)$ which is depicted in Figure~\ref{fig:prelim-example}(b). We have that $\sigma(\ell) = m$ since the opinion $z_m=-5$ of the middle player is closer  to the belief $s_\ell=-10$ of the left player than the opinion $z_r=4$ of the right player. Therefore, the cost of the left player is $\cost_\ell(\zz,\sss) = \max\{|-10+10|,|-10+5|\} = 5$. Similarly, the neighbors of the middle and right players are $\sigma(m) = r$ and $\sigma(r) = m$, while their costs are $\cost_m(\zz,\sss) = \max\{2+5,4+5\} = 9$ and $\cost_r(\zz,\sss) = \max\{5-4,4+5\} = 9$, respectively. The social cost is $\SC(\zz,\sss) = 23$.

Now, consider the alternative pure Nash equilibrium opinion vector $\zz' = (-3.5,3,4)$ which is depicted in Figure~\ref{fig:prelim-example}(c). Observe that even though $\zz' \neq \zz$, each player has the same neighbor as in $\zz$ and no player has an incentive to deviate in order to decrease her cost. Indeed, let us focus on the middle player for whom it is $\sigma(m) = r$. Her opinion is in the middle of the interval defined by her belief $s_m = 3$ and the opinion $z_r' = 5$ of the right player. Hence, this opinion minimizes her cost by minimizing the maximum between the distance from her belief and the distance from the opinion of the right player. It is easy to verify that the same holds for the left and right players. The player costs are now $6.5$, $1$, and $1$, respectively, yielding a social cost of $8.5$.
\end{example}

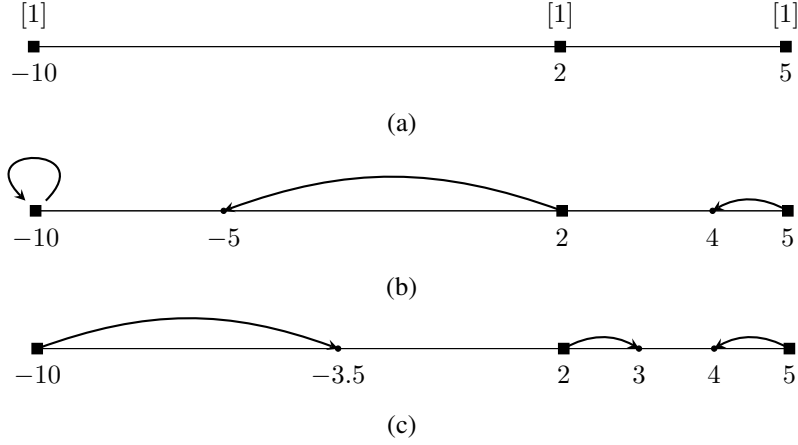
\begin{figure}[h]
\centering
\begin{subfigure}{\textwidth}
\centering
\begin{tikzpicture}[scale = 1]
  \filldraw (0,0) -- (10,0);

  \filldraw (0,-0.1) node[align=center, below] {\small $-10$};
  \filldraw (0,0.1) node[align=center, above] {\small $[1]$};

  \filldraw (7,-0.1) node[align=center, below] {\small $2$};
  \filldraw (7,0.1) node[align=center, above] {\small $[1]$};

  \filldraw (10,-0.1) node[align=center, below] {\small $5$};
  \filldraw (10,0.1) node[align=center, above] {\small $[1]$};

  \filldraw ([xshift=-2pt,yshift=-2pt]0,0) rectangle ++(4pt,4pt);
  \filldraw ([xshift=-2pt,yshift=-2pt]7,0) rectangle ++(4pt,4pt);
  \filldraw ([xshift=-2pt,yshift=-2pt]10,0) rectangle ++(4pt,4pt);
\end{tikzpicture}
\caption{ \ }
\end{subfigure}
\begin{subfigure}{\textwidth}
\centering
\hspace*{-0.5cm}\begin{tikzpicture}[scale = 1]
  \filldraw (0,0) -- (10,0);

  \filldraw (0,-0.1) node[align=center, below] {\small $-10$};

  \filldraw (7,-0.1) node[align=center, below] {\small $2$};

  \filldraw (10,-0.1) node[align=center, below] {\small $5$};

  \filldraw ([xshift=-2pt,yshift=-2pt]0,0) rectangle ++(4pt,4pt);
  \filldraw ([xshift=-2pt,yshift=-2pt]7,0) rectangle ++(4pt,4pt);
  \filldraw ([xshift=-2pt,yshift=-2pt]10,0) rectangle ++(4pt,4pt);

  \filldraw (2.5,0) circle(1pt);
  \filldraw (2.5,-0.1) node[align=center,below] {\small $-5$};

  \filldraw (9,0) circle(1pt);
  \filldraw (9,-0.1) node[align=center,below] {\small $4$};
	
  \node at (0,0) (zero){};

  \path[-stealth,thick] (zero) edge[out=45,in=135,looseness=10,relative=false] (zero);
  \draw[-stealth,thick] (7,0) to[bend right=20] (2.5,0);
  \draw[-stealth,thick] (10,0) to[bend right=30] (9,0);
\end{tikzpicture}
\caption{ \ }
\end{subfigure}
\begin{subfigure}{\textwidth}
\centering
\begin{tikzpicture}[scale = 1]
  \filldraw (0,0) -- (10,0);

  \filldraw (0,-0.1) node[align=center, below] {\small $-10$};

  \filldraw (7,-0.1) node[align=center, below] {\small $2$};

  \filldraw (10,-0.1) node[align=center, below] {\small $5$};

  \filldraw ([xshift=-2pt,yshift=-2pt]0,0) rectangle ++(4pt,4pt);
  \filldraw ([xshift=-2pt,yshift=-2pt]7,0) rectangle ++(4pt,4pt);
  \filldraw ([xshift=-2pt,yshift=-2pt]10,0) rectangle ++(4pt,4pt);

  \filldraw (4,0) circle(1pt);
  \filldraw (4,-0.1) node[align=center,below] {\small $-3.5$};

  \filldraw (8,0) circle(1pt);
  \filldraw (8,-0.1) node[align=center,below] {\small $3$};

  \filldraw (9,0) circle(1pt);
  \filldraw (9,-0.1) node[align=center,below] {\small $4$};
	
  \draw[-stealth,thick] (0,0) to[bend left=20] (4,0);
  \draw[-stealth,thick] (7,0) to[bend left=30] (8,0);
  \draw[-stealth,thick] (10,0) to[bend right=30] (9,0);

\end{tikzpicture}
\caption{ \ }
\end{subfigure}
\caption{The game examined in Example~\ref{ex:prelim-example}. (a) Illustration of the belief vector $\sss = (-10, 2, 5)$. The black squares correspond to player beliefs. The notation $[x]$ is used to denote the number of players that have the same beliefs; here we have only one player per belief. (b) Illustration of the opinion vector $\zz=(-10,-5,4)$. The dots correspond to player opinions and each arrow connects the belief of a player to her opinion. (c) Illustration of the equilibrium opinion vector $\zz'=(-3.5,3,4)$.}
\label{fig:prelim-example}
\end{figure}

\section{Some properties about equilibria}\label{sec:properties}
We devote this section to proving several some interesting properties of pure Nash equilibria; these will be useful in the following. The first one is obvious due to the definition of the cost function.

\begin{lemma}\label{lem:equilibrium}
In any pure Nash equilibrium $\zz$ of a $k$-COF game with belief vector $\sss$, the opinion of any player $i$ lies in the middle of the interval $I_i(\zz, \sss)$.
\end{lemma}

The next lemma allows us to argue about the order of player opinions in a pure Nash equilibrium $\zz$.

\begin{lemma}\label{lem:mon}
In any pure Nash equilibrium $\zz$ of a $k$-COF game with belief vector $\sss$, it holds that $z_i\leq z_{i+1}$ for any $i \in [n-1]$ such that $s_i < s_{i+1}$.
\end{lemma}

\begin{proof}
For the sake of contradiction, let us assume that $z_{i+1}<z_i$ for a pair of players $i$ and $i+1$ with $s_i < s_{i+1}$. Then, it cannot be the case that the leftmost endpoint of the interval $I_i(\zz,\sss)$ of player $i$ is at the left of (or coincides with) the leftmost endpoint of interval $I_{i+1}(\zz,\sss)$ of player $i+1$ and the rightmost endpoint of $I_i(\zz,\sss)$ is at the left of (or coincides with) the rightmost endpoint of $I_{i+1}(\zz,\sss)$. In other words, it cannot be the case that $\min\{s_i,z_{\ell(i)}\}\leq \min\{s_{i+1},z_{\ell(i+1)}\}$ and $\max\{s_i,z_{r(i)}\} \leq \max\{s_{i+1},z_{r(i+1)}\}$ hold simultaneously. Since, by Lemma \ref{lem:equilibrium}, points $z_i$ and $z_{i+1}$ lie in the middle of the corresponding intervals, we would have $z_i\leq z_{i+1}$, contradicting our assumption.

So, at least one of the two inequalities between the interval endpoints above must not hold. In the following, we assume that $\min\{s_i,z_{\ell(i)}\}>\min\{s_{i+1},z_{\ell(i+1)}\}$ (the case where $\max\{s_i,z_{r(i)}\} >\max\{s_{i+1},z_{r(i+1)}\}$ is symmetric).
This assumption implies that $z_{\ell(i+1)}<s_i < s_{i+1}$ (i.e., $\min\{s_{i+1},z_{\ell(i+1)}\} = z_{\ell(i+1)}$), and, subsequently, that $z_{\ell(i+1)}<z_{\ell(i)}$. In words, player $\ell(i+1)$ does not belong to interval $I_i(\zz,\sss)$. Furthermore, since $z_{\ell(i+1)} < s_{i+1}$, and as (by Lemma~\ref{lem:equilibrium}) $z_{i+1}$ lies in the middle of $I_{i+1}(\zz,\sss)$, we also have that the leftmost endpoint of interval $I_{i+1}(\zz,\sss)$ cannot belong to player $i+1$, i.e., $\ell(i+1)\not=i+1$. An example of the relative ordering of points (beliefs and opinions), after assuming that $z_{i+1} < z_i$ and $\min\{s_i,z_{\ell(i)}\}>\min\{s_{i+1},z_{\ell(i+1)}\}$ is depicted in Figure \ref{fig:example-Lem2}.

\begin{figure}[h!]
\centering
\begin{tikzpicture}[xscale=1, yscale=1]
  \filldraw[thick] (0,0) -- (10,0);
  \filldraw(0,-0.1) node[align=center, below] {\small $z_{\ell(i+1)}$};
  \filldraw(2,-0.1) node[align=center, below] {\small $z_{\ell(i)}$};
  \filldraw(4,-0.1) node[align=center, below] {\small $s_i$};
  \filldraw(6,-0.1) node[align=center, below] {\small $z_{i+1}$};
  \filldraw(8,-0.1) node[align=center, below] {\small $z_i$};
  \filldraw(10,-0.1) node[align=center, below] {\small $s_{i+1}$};

  \filldraw (0,0) circle(1pt);
  \filldraw (2,0) circle(1pt);
  \filldraw (6,0) circle(1pt);
  \filldraw (8,0) circle(1pt);

  \filldraw ([xshift=-2pt,yshift=-2pt]4,0) rectangle ++(4pt,4pt);
  \filldraw ([xshift=-2pt,yshift=-2pt]10,0) rectangle ++(4pt,4pt);

  \draw[-stealth,thick] (4,0) to[bend left=25] (8,0);
  \draw[-stealth,thick] (10,0) to[bend right=25] (6,0);

  \draw[-stealth,thick] (0,0.5) to (0,0);
  \draw[-stealth,thick] (2,0.5) to (2,0);

  \node at (0,0) (zleft){};
  \node at (2,0) (zright){};

  \draw[thick] ([xshift=0em,yshift=5em]zleft.center) -- ([xshift=1.5em,yshift=5em]zleft.center);
  \draw[thick] ([xshift=0em,yshift=4.7em]zleft.center) -- ([xshift=0em,yshift=5.3em]zleft.center);
  \node[align=center] at ([xshift=3.4em,yshift=5em]zleft.center) {\scriptsize $I_{i+1}(\zz,\sss)$};
  \draw[thick] ([xshift=5.3em,yshift=5em]zleft.center) -- ([xshift=6.3em,yshift=5em]zleft.center);
  \node[align=center] at ([xshift=7em,yshift=5em]zleft.center) {\scriptsize $\cdots$};

  \draw[thick] ([xshift=0em,yshift=3em]zright.center) -- ([xshift=1.5em,yshift=3em]zright.center);
  \draw[thick] ([xshift=0em,yshift=2.7em]zright.center) -- ([xshift=0em,yshift=3.3em]zright.center);
  \node[align=center] at ([xshift=3em,yshift=3em]zright.center) {\scriptsize $I_i(\zz,\sss)$};
  \draw[thick] ([xshift=4.5em,yshift=3em]zright.center) -- ([xshift=5.5em,yshift=3em]zright.center);
  \node[align=center] at ([xshift=6.2em,yshift=3em]zright.center) {\scriptsize $\cdots$};

  \draw[dashed] ([xshift=0em,yshift=5em]zleft.center) -- (zleft.center);
  \draw[dashed] ([xshift=0em,yshift=3em]zright.center) -- (zright.center);
\end{tikzpicture}
\caption{An example of the argument used in the proof of Lemma \ref{lem:mon}.}
\label{fig:example-Lem2}
\end{figure}
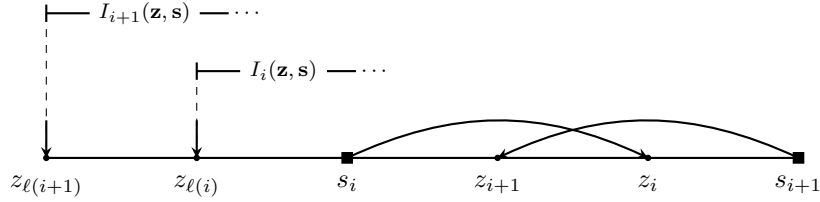

Since $\ell(i+1)$ does not belong to $I_i(\zz,\sss)$, there are at least $k$ players different than $\ell(i+1)$ and $i$ that have opinions at distance at most $s_i-z_{\ell(i+1)}$ from belief $s_i$. Since $s_i < s_{i+1}$ and $z_{\ell(i+1)} < z_{\ell(i)}$, all these players are also at distance strictly less than $s_{i+1}-z_{\ell(i+1)}$ from belief $s_{i+1}$. This contradicts the fact that the opinion of player $\ell(i+1)$ is among the $k$ closest opinions to $s_{i+1}$.
\end{proof}

In the following, in any pure Nash equibrium $\zz$, we assume that $z_i\leq z_{i+1}$ for any $i \in [n-1]$. This follows by Lemma \ref{lem:mon} when $s_i < s_{i+1}$ and by a convention for the identities of players with identical belief.

In addition to the ordering of opinions in a pure Nash equilibrium, we can also specify the range of neighborhoods (in Lemma~\ref{lem:nei}) and opinions (in Lemma~\ref{lem:betw}).

\begin{lemma}\label{lem:nei}
Let $\zz$ be a pure Nash equilibrium of a $k$-COF game with belief vector $\sss$. Then, for each player $i$, there exists $j$ with $i-k\leq j\leq i$ such that $I_i(\zz,\sss)$ is the shortest interval that contains the opinions $z_j, z_{j+1}, ..., z_{j+k}$ and belief $s_i$.
\end{lemma}

\begin{proof}
If $I_i(\zz,\sss)$ consists of a single point, the lemma follows trivially by the definition of the neighborhood and Lemma \ref{lem:mon} since at least $k+1$ consecutive players including $i$ should have opinions in $I_i(\zz,\sss)$. Otherwise, by Lemma \ref{lem:mon}, the lemma is true if there is at most one opinion in each of the left and the right boundary of $I_i(\zz,\sss)$; in this case, there are exactly $k+1$ consecutive players including player $i$ with opinions in $I_i(\zz,\sss)$.

In the following, we handle the subtleties that may arise due to tie-breaking at the boundaries of $I_i(\zz,\sss)$. Let $Y_\ell$ and $Y_r$ be the set of players with opinions at the leftmost and the rightmost point of $I_i(\zz,\sss)$, respectively. From Lemma \ref{lem:equilibrium}, player $i$ belongs neither to $Y_\ell$ nor to $Y_r$. Now consider the following set of players: the $|Y_\ell\cap N_i(\zz,\sss)|$ players with highest indices from $Y_\ell$, the $|Y_r\cap N_i(\zz,\sss)|$ players with lowest indices from $Y_r$ and all players with opinions that lie strictly in $I_i(\zz,\sss)$. Due to the definition of $N_i(\zz,\sss)$ and by Lemma \ref{lem:mon}, there are $k+1$ players in this set, including player $i$, with consecutive indices.
\end{proof}
\noindent In the following, irrespectively of how ties are actually resolved, we assume that $N_i(\zz,\sss)\cup \{i\}$ consists of $k+1$ players with consecutive indices. This does not affect the cost of player $i$ at equilibrium in the proofs of our upper bounds (since, by Lemma \ref{lem:nei}, the interval defined is exactly the same), while our lower bound constructions are defined carefully so that the results hold no matter how ties are actually resolved.

\begin{lemma}\label{lem:betw}
Let $\zz$ be a pure Nash equilibrium of a $k$-COF game with belief vector $\sss$. Then, for each player $i$, it holds that $s_{\ell(i)} \leq z_i\leq s_{r(i)}$.
\end{lemma}

\begin{proof}
Since $N_i(\zz,\sss) \cup \{i\}$ consists of $k+1$ players with consecutive indices, we have that $s_{\ell(i)} \leq s_i \leq s_{r(i)}$.
For the sake of contradiction, let us assume that $s_{\ell(i)}\leq s_{r(i)}<z_i$ for some player $i$ (the case where $z_i$ lies at the left of $s_{\ell(i)}$ is symmetric). Since $s_{r(i)} < s_i$ and as $z_i$ is at the middle of $I_i(\zz, \sss)$, it holds that $z_{r(i)}>z_i$ (i.e., $r(i) \neq i$). Also, since $z_{r(i)}>z_i>s_{r(i)}$, and because $z_{r(i)}$ is in the middle of $I_{r(i)}(\zz, \sss)$, it holds that $z_{r(r(i))} > z_{r(i)}$ and, by Lemma \ref{lem:mon}, $r(r(i))>r(i)$; see Figure \ref{fig:example-Lem3} for an example of the relative ordering of points (beliefs and opinions) when assuming that $s_{r(i)} < z_i$.

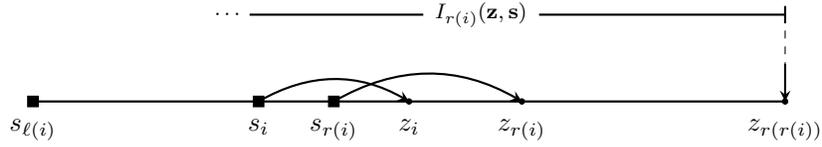
\begin{figure}[h!]
\centering
\begin{tikzpicture}[xscale=1, yscale=1]
  \filldraw[thick] (0,0) -- (10,0);
  \filldraw (0,-0.1) node[align=center, below] {\small $s_{\ell(i)}$};
  \filldraw (3,-0.1) node[align=center, below] {\small $s_i$};
  \filldraw (4,-0.1) node[align=center, below] {\small $s_{r(i)}$};
  \filldraw (5,-0.1) node[align=center, below] {\small $z_i$};
  \filldraw (6.5,-0.1) node[align=center, below] {\small $z_{r(i)}$};
  \filldraw (10,-0.1) node[align=center, below] {\small $z_{r(r(i))}$};

  \filldraw (5,0) circle(1pt);
  \filldraw (6.5,0) circle(1pt);
  \filldraw (10,0) circle(1pt);

  \filldraw ([xshift=-2pt,yshift=-2pt]0,0) rectangle ++(4pt,4pt);
  \filldraw ([xshift=-2pt,yshift=-2pt]3,0) rectangle ++(4pt,4pt);
  \filldraw ([xshift=-2pt,yshift=-2pt]4,0) rectangle ++(4pt,4pt);

  \draw[-stealth,thick] (3,0) to[bend left=30] (5,0);
  \draw[-stealth,thick] (4,0) to[bend left=30] (6.5,0);

  \draw[-stealth,thick] (10,0.5) to (10,0);

  \node at (10,0) (z){};

  \draw[thick] ([xshift=0em,yshift=3em]z.center) -- ([xshift=-8.5em,yshift=3em]z.center);
  \draw[thick] ([xshift=0em,yshift=2.7em]z.center) -- ([xshift=0em,yshift=3.3em]z.center);
  \node[align=center] at ([xshift=-10.5em,yshift=3em]z.center) {\scriptsize $I_{r(i)}(\zz,\sss)$};
  \draw[thick] ([xshift=-12.5em,yshift=3em]z.center) -- ([xshift=-18.5em,yshift=3em]z.center);
  \node[align=center] at ([xshift=-19.2em,yshift=3em]z.center) {\scriptsize $\cdots$};

  \draw[dashed] ([xshift=0em,yshift=3em]z.center) -- (z.center);
\end{tikzpicture}
\caption{An example of the argument used in the proof of Lemma \ref{lem:betw}.}
\label{fig:example-Lem3}
\end{figure}

We now claim that $\ell(i) \notin N_{r(i)}(\zz, \sss)$. Assume otherwise that $\ell(i)\in N_{r(i)}(\zz, \sss)$. By definition, $r(r(i))\in N_{r(i)}(\zz, \sss)$. Then, Lemma \ref{lem:mon} implies that any player $j$, different than $r(i)$, with $\ell(i)<j<r(r(i))$ is also in $N_{r(i)}(\zz, \sss)$. Hence,
$N_{r(i)}(\zz,\sss)$ contains at least the $k-1$ players in $N_i(\zz,\sss)\setminus\{r(i)\}$, as well as players $i$ and $r(r(i))$. This, however, contradicts the fact that $|N_{r(i)}(\zz, \sss)| = k$. Therefore, player $\ell(i)$ is not among the $k$ nearest neighbors of $r(i)$.

So, we obtain that
\begin{align*}
z_{r(r(i))} - s_{r(i)} &> z_{r(i)} - s_{r(i)} > z_{r(i)} - z_i = z_i - \min\{s_i, z_{\ell(i)}\} \\
&> s_{r(i)} - \min\{s_i, z_{\ell(i)}\} \geq s_{r(i)} - z_{\ell(i)}.
\end{align*}
If $z_{\ell(i)} > s_{r(i)}$ (i.e, $z_{\ell(i)}$ is at the right of $s_{r(i)}$), then since, by Lemma \ref{lem:mon}, $z_{\ell(i)}\leq z_{r(r(i))}$ and $r(r(i)) \in N_{r(i)}(\zz, \sss)$, we obtain that $\ell(i) \in N_{r(i)}(\zz, \sss)$ as well; a contradiction. Otherwise, the above inequality yields  that $z_{r(r(i))} - s_{r(i)} > s_{r(i)} - z_{\ell(i)}\geq 0$ (i.e., the distance of $s_{r(i)}$ from $z_{r(r(i))}$ is strictly higher than the distance of $s_{r(i)}$ from $z_{\ell(i)}$), and, again, we obtain a contradiction to the fact that $\ell(i)\notin N_{r(i)}(\zz, \sss)$ and $r(r(i)) \in N_{r(i)}(\zz, \sss)$.
\end{proof}


\section{Existence and quality of equilibria}\label{sec:existence}
Our first technical contribution is a negative statement: pure Nash equilibria may not exist for any $k$ (Theorem \ref{thm:no-pure}). Then, we show that even in $k$-COF games that admit pure Nash equilibria, the best equilibrium may be inefficient; in other words, the price of stability is strictly greater than $1$ for any value of $k$, and, actually, depends linearly on $k$. These results appear in Theorems \ref{thm:pos-k}, \ref{thm:pos}, and \ref{thm:pos-2}.

\subsection{Existence of equilibria}
We begin with a technical lemma. The lemma essentially presents necessary conditions so that a particular set of neighborhoods, and corresponding intervals, may coexist in a pure Nash equilibrium.

\begin{lemma}\label{lem:small-chain}
Consider a $k$-COF game and any three players $a, b, c$ with beliefs $s_a\leq s_b\leq s_c$, respectively. For any pure Nash equilibrium $\zz$ where $I_a(\zz, \sss) = [s_a, z_b]$, $I_b(\zz, \sss) = [s_b, z_c]$ and $I_c(\zz, \sss) = [z_b, s_c]$, it must hold that $s_b\geq\frac{3s_a+5s_c}{8}$, while for any pure Nash equilibrium $\zz$ where $I_a(\zz, \sss) = [s_a, z_b]$, $I_b(\zz, \sss) = [z_a, s_b]$ and $I_c(\zz, \sss) = [z_b, s_c]$, it must hold that $s_b\leq\frac{5s_a+3s_c}{8}$.
\end{lemma}

\begin{proof}
It suffices to prove the first case; the second case is symmetric. Since $I_b(\zz, \sss) = [s_b, z_c]$ and $I_c(\zz, \sss) = [z_b, s_c]$, by Lemma \ref{lem:equilibrium} it holds that $z_b = (s_b+z_c)/2$ and $z_c = (z_b+s_c)/2$ which yield that $z_b = s_b+\frac{s_c-s_b}{3}$ and $z_c = s_b+\frac{2(s_c-s_b)}{3}$. Hence, we obtain that
\begin{align}\label{eq:small-chain-1}
z_c - s_b = \frac{2(s_c-s_b)}{3}.
\end{align}
Similarly, since $I_a(\zz, \sss) = [s_a, z_b]$, it holds that $z_a = \frac{s_a+z_b}{2}= \frac{3s_a+2s_b+s_c}{6}$ and, therefore, we obtain that
\begin{align}\label{eq:small-chain-2}
s_b-z_a = \frac{-3s_a+4s_b-s_c}{6}.
\end{align}
Since $I_b(\zz, \sss) = [s_b, z_c]$, we have that $a\notin N_b(\zz, \sss)$ and, subsequently, that $z_c-s_b\leq s_b-z_a$ which, together with (\ref{eq:small-chain-1}) and (\ref{eq:small-chain-2}), yields that $s_b\geq\frac{3s_a+5s_c}{8}$ as desired.
\end{proof}

The proof of the next theorem is inspired by a construction of \citet{BGM13} and exploits Lemma~\ref{lem:small-chain}.

\begin{theorem}\label{thm:no-pure}
For any $k$, there exists a $k$-COF game with no pure Nash equilibria.
\end{theorem}

\begin{proof}
Consider a $k$-COF game with $2k+1$ players partitioned into three sets called $L$, $M$, and $R$, where $L$ and $R$ each contain $k$ players, while $M = \{m\}$ is a singleton.  We set $s_i = 0$ for each $i \in L$, $s_i = 2$ for each $i \in R$, while $s_m = 1-\epsilon$, where $\epsilon < 1/4$ is an arbitrarily small positive constant.

Let us assume that there exists a pure Nash equilibrium $\zz$. Then, clearly, for any $i \in L$ it must hold that $N_i(\zz, \sss) = L\setminus \{i\} \cup \{m\}$, and, therefore, $I_i(\zz, \sss) = [0, z_m]$. Similarly, for any $i \in R$ we have $N_i(\zz, \sss) = R \setminus \{i\} \cup \{m\}$, and $I_i(\zz, \sss) = [z_m, 2]$. Now, concerning player $m$, if all her neighbors are in $L$, then, it holds that $I_m(\zz, \sss) = [z_i, s_m]$ for some $i\in L$. But then, observe that even though the intervals defined above exhibit the structure described in Lemma \ref{lem:small-chain}, the belief vector $\sss$ does not satisfy the corresponding necessary conditions of that lemma as $1-\epsilon > 3/4$; hence, $\zz$ is not a pure Nash equilibrium. The same reasoning applies in case all of $m$'s neighbors are in $R$.

It remains to consider the case where $m$ has at least one neighbor in each of $L$ and $R$. By the definition of $I_i(\zz, \sss)$ for $i \in L \cup R$, as stated above, Lemma \ref{lem:equilibrium} implies that $z_i = z_m/2$ for any $i\in L$, while $z_i = 1+z_m/2$ for any $i\in R$. Then, Lemma \ref{lem:betw} implies that $z_m/2\leq s_m = 1-\epsilon$ and $1+z_m/2 \geq s_m$, and, consequently, $I_m(\zz, \sss) = [z_m/2, 1+z_m/2]$. Again, by Lemma \ref{lem:equilibrium} we have that $z_m = \frac{z_m/2+1+z_m/2}{2}$, i.e., $z_m = 1$. But then, we obtain $z_i= 1/2$ for any $i\in L$ and $z_i = 3/2$ for any $i\in R$, which implies that all $k$ players in $L$ are strictly closer to $s_m$ than any player in $R$; this contradicts the assumption that $m$ has neighbors in both $L$ and $R$.
\end{proof}

An example of the construction used in the proof of Theorem 6 is presented in Figure~\ref{fig:no-pure}.

\begin{figure}[h]
\begin{subfigure}{\textwidth}
\centering
\begin{tikzpicture}[xscale=1, yscale=1]
  \filldraw (0,0) -- (4.5,0) -- (10,0);

  \filldraw (0,-0.1) node[align=center, below] {\small $0$};
  \filldraw (0,0.1) node[align=center, above] {\small $[k]$};

  \filldraw (4.5,-0.1) node[align=center, below] {\small $1-\epsilon$};
  \filldraw (4.5,0.1) node[align=center, above] {\small $[1]$};

  \filldraw (10,-0.1) node[align=center, below] {\small $2$};
  \filldraw (10,0.1) node[align=center, above] {\small $[k]$};

  \filldraw ([xshift=-2pt,yshift=-2pt]0,0) rectangle ++(4pt,4pt);
  \filldraw ([xshift=-2pt,yshift=-2pt]4.5,0) rectangle ++(4pt,4pt);
  \filldraw ([xshift=-2pt,yshift=-2pt]10,0) rectangle ++(4pt,4pt);
\end{tikzpicture}
\caption{ \ }
\end{subfigure}
\begin{subfigure}{\textwidth}
\centering
\begin{tikzpicture}[xscale=1, yscale=1]
  \filldraw (0,0) -- (10,0);

  \filldraw (0,-0.1) node[align=center, below] {\small $0$};

  \filldraw (4.5,-0.1) node[align=center, below] {\small $1-\epsilon$};

  \filldraw (10,-0.1) node[align=center, below] {\small $2$};

  \filldraw ([xshift=-2pt,yshift=-2pt]0,0) rectangle ++(4pt,4pt);
  \filldraw ([xshift=-2pt,yshift=-2pt]4.5,0) rectangle ++(4pt,4pt);
  \filldraw ([xshift=-2pt,yshift=-2pt]10,0) rectangle ++(4pt,4pt);

  \filldraw (2.6125,0) circle(1pt);
  \filldraw (2.6125,-0.1) node[align=center,below] {\small $x/2$};

  \filldraw (5.25,0) circle(1pt);
  \filldraw (5.25,-0.15) node[align=center,below] {\small $x$};

  \filldraw (7.625,0) circle(1pt);
  \filldraw (7.625,-0.05) node[align=center,below] {\small $1+x/2$};

  \draw[-stealth,thick] (0,0) to[bend left=20] (2.6125,0);
  \draw[-stealth,thick] (4.5,0) to[bend left=30] (5.25,0);
  \draw[-stealth,thick] (10,0) to[bend right=20] (7.625,0);

\end{tikzpicture}
\caption{ \ }
\end{subfigure}
\caption{(a) The $k$-COF game considered in the proof of Theorem~\ref{thm:no-pure} where the $k$ players of set $L$ have belief $0$, player $m$ has $s_m = 1-\epsilon$ and the $k$ players of set $R$ have belief $2$. (b) Lemma \ref{lem:small-chain} implies that there is no pure Nash equilibrium where $m$ has neighbors in strictly one of $L$, $R$. In the remaining case, it must hold that $x = 1$, but then all players in $L$ are strictly closer to $s_m$ than any player in $R$.}
\label{fig:no-pure}
\end{figure}
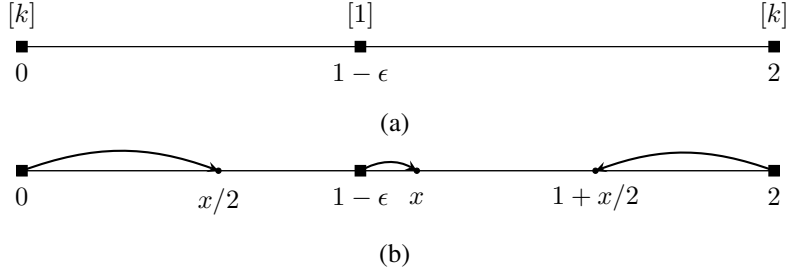

\subsection{Price of stability}
We will now prove that the price of stability of $k$-COF games is strictly higher than $1$, i.e., there exist games without any efficient pure Nash equilibria (even when they exist). In particular, for any value of $k$ we show that there exist rather simple games with price of stability in $\Omega(k)$.

\begin{theorem}\label{thm:pos-k}
The price of stability of $k$-COF games, for $k\geq 3$, is at least $(k+1)/3$.
\end{theorem}
\begin{proof}
Consider a $k$-COF game with $k+1$ players, where $k$ of them have belief $0$, while the remaining one has belief $1$. Let $\tilde{\zz}$ be the opinion vector where each player has opinion $0$. Clearly, $\SC(\tilde{\zz}, \sss) = 1$, and, hence the optimal social cost is at most $1$.

Now, consider any pure Nash equilibrium $\zz$. Since, there are $k+1$ players, the neighborhood of each player includes all remaining ones. Let $x$ be the opinion that the player with belief $1$ expresses at $\zz$. By Lemma \ref{lem:betw}, we have that $x \in [0,1]$, and by Lemma \ref{lem:equilibrium}, we have that all remaining players must have opinion $x/2$. Therefore, again by Lemma \ref{lem:equilibrium}, $x$ must satisfy the equation $x = (1+x/2)/2$, i.e., $x=2/3$. Therefore, there exists a single pure Nash equilibrium $\zz$ where all players with belief $0$ have opinion $1/3$ and the single player with belief $1$ has opinion $2/3$, and we obtain $\SC(\zz) = (k+1)/3$ which implies the theorem.
\end{proof}

Clearly, the above result states the inefficiency of the best pure Nash equilibrium only when $k\geq 3$. For the remaining cases where $k\in \{1,2\}$ we present slightly more complicated instances, where the proofs rely on Lemma \ref{lem:small-chain}. Recall that, for $1$-COF games, $\sigma(i)$ denotes the single neighbor of player $i$.

\begin{theorem}\label{thm:pos}
The price of stability of $1$-COF games is at least $17/15$.
\end{theorem}
\begin{proof}
We use the following $1$-COF game with six players and belief vector
$$\sss = (0, 5-3\lambda, 8, 15, 18+3\lambda, 23),$$
where $\lambda\in(0,1/4)$.

Consider the opinion vector
$$\tilde{\zz} = (3-\lambda, 6-2\lambda, 7-6\lambda, 16+6\lambda, 17+2\lambda, 20+\lambda).$$
It can be easily seen that it has social cost
$\SC(\tilde{\zz}, \sss) = 10+12\lambda$. So, clearly, $\SC(\zz^*, \sss)\leq 10+12\lambda$ for any optimal opinion vector $\zz^*$.

Now, consider the opinion vector
$$\zz = \left(\frac{5-3\lambda}{3}, \frac{10-6\lambda}{3}, \frac{31}{3}, \frac{38}{3}, \frac{59+6\lambda}{3}, \frac{64+3\lambda}{3}\right)$$
with social cost $\SC(\zz, \sss) = 34/3-4\lambda$. It is not hard to verify (by showing, as Lemma \ref{lem:equilibrium} requires, that each opinion lies in the middle of its player's interval) that $\zz$ is a pure Nash equilibrium; we argue that this equilibrium is unique.

We claim that, by Lemma \ref{lem:small-chain}, there cannot be a pure Nash equilibrium where both $\sigma(j-1) = j$ and $\sigma(j+1) = j$ for any $j\in \{2,5\}$. To see this, assume otherwise and note that the corresponding intervals satisfy the conditions of the lemma. However, by observing the belief vector $\sss$, it holds that $\frac{5s_{j-1}+3s_{j+1}}{8}<s_j<\frac{3s_{j-1}+5s_{j+1}}{8}$, for $j\in\{2,5\}$, i.e., $\sss$ does not satisfy the conditions of Lemma \ref{lem:small-chain}; this contradicts our original assumption.

The above observation, together with Lemma \ref{lem:mon}, implies that $\sigma(1) = 2$, $\sigma(3) =4$, $\sigma(4)=3$ and $\sigma(6)=5$ in any equilibrium. This leaves only $\sigma(2)\in\{1,3\}$ and $\sigma(5)\in\{4,6\}$ undefined.

Consider an equilibrium $\zz'$ with $\sigma(2) = 3$; the case $\sigma(5)=4$ is symmetric. Since $\sigma(3) = 4$, Lemma \ref{lem:betw} implies that $z'_3>s_3=8$ and, hence
\begin{align}\label{eq:pos-1}
z'_3-s_2>3+3\lambda.\end{align}
Since $\sigma(1) = 2$, $\sigma(2) =3$ and $z'_1 = \frac{s_1+z'_2}{2}$, Lemma \ref{lem:betw} implies that $z'_2>s_2$ and we obtain that $z'_1>\frac{5-3\lambda}{2}$ and, hence,
\begin{align}\label{eq:pos-2}
s_2-z'_1<\frac{5-3\lambda}{2}.\end{align}
By inequalities (\ref{eq:pos-1}) and (\ref{eq:pos-2}), we get $z'_3-s_2>s_2-z'_1$, which contradicts our assumption that $\sigma(2) = 3$. So, it must hold that $\sigma(2)=1$ (and, respectively, $\sigma(5) =6$) which implies that $\zz$ is the unique pure Nash equilibrium.

We conclude that the price of stability is lower-bounded by
$$\frac{\SC(\zz, \sss)}{\SC(\zz^*, \sss)}=\frac{34/3-4\lambda}{10+12\lambda},$$
and the theorem follows by taking $\lambda$ to be arbitrarily close to $0$.
\end{proof}

\begin{theorem}\label{thm:pos-2}
The price of stability of $2$-COF games is at least $8/7$.
\end{theorem}
\begin{proof}
Consider a $2$-COF game with four players $a$, $b$, $c$, and $d$, with belief vector $\sss = (0, 1, 1, 2)$. Let $\tilde{\zz} = (1,1,1,3/2)$ be an opinion vector and observe that $SC(\tilde{\zz}, \sss) = 3/2$; note that $\tilde{\zz}$ is not a pure Nash equilibrium as player $a$ has an incentive to deviate. Clearly, the optimal social cost is at most $3/2$.

Now consider any pure Nash equilibrium $\zz$. By the structural properties of equilibria, $N_a(\zz, \sss) = N_d(\zz, \sss) = \{b,c\}$, while $b \in N_c(\zz, \sss)$ and $c \in N_b(\zz, \sss)$. It remains to argue about the second neighbor of $b$ and $c$. We distinguish between two cases depending on whether $b$ and $c$ have a common second neighbor in $\{a, d\}$ or not.

In the first case, let $a$ be the common neighbor; the case where $d$ is that neighbor is symmetric. By Lemma \ref{lem:equilibrium}, we have that $z_b = z_c = (1+z_a)/2$. Then, we have that $I_a(\zz, \sss) = [0, z_b]$, $I_b(\zz, \sss) = [z_a, 1]$, and $I_d(\zz, \sss) = [z_b, 2]$. Note that by applying Lemma \ref{lem:small-chain} on players $a$, $b$, and $d$, we obtain a contradiction to the fact that $\zz$ is a pure Nash equilibrium.

In the second case, without of loss of generality, let $N_b(\zz, \sss) = \{a, c\}$ and $N_c(\zz, \sss) = \{b, d\}$ which, by Lemma \ref{lem:betw}, imply that $z_b \in [0, 1]$ and $z_c \in [1, 2]$. Then, Lemma \ref{lem:equilibrium} yields $z_a = z_c/2$, $z_b = (z_a+z_c)/2$, $z_c = (z_b+z_d)/2$, and $z_d = 1+z_b/2$. By solving this system of equations, we obtain that $\zz = (4/7, 6/7, 8/7, 10/7)$ and, hence, $\SC(\zz) = 12/7$.
\end{proof}

\section{Complexity of equilibria}\label{sec:algorithm}
In this section we focus entirely on $1$-COF games. We present a polynomial-time algorithm that determines whether such a game admits pure Nash equilibria, and, in case it does, allows us to compute the best and worst pure Nash equilibrium with respect to the social cost. We do so by establishing a correspondence between pure Nash equilibria and source-sink paths in a suitably defined directed acyclic graph. See Example~\ref{ex:algo-example} below for an instance execution of the following procedure.

Assume that we are given neighborhood information according to which each player $i$ has either player $i-1$ or player $i+1$ as neighbor. From Lemma \ref{lem:nei}, such a neighborhood structure is necessary in a pure Nash equilibrium. We claim that this information is enough in order to decide whether there is a consistent opinion vector that is a pure Nash equilibrium or not. All we have to do is to use Lemma \ref{lem:equilibrium} and obtain $n$ equations that relate the opinion of each player to her belief and her neighbor's opinion. These equations have a unique solution which can then be verified whether it indeed satisfies the neighborhood conditions or not. So, the main idea of our algorithm is to cleverly search among all possible neighborhood structures that are not excluded by Lemma \ref{lem:nei} for one that defines a pure Nash equilibrium.

For integers $1\leq a\leq b<c\leq n$, let us define the {\em segment} $C(a,b,c)$ to be the set of players $\{a, a+1, ..., c\}$ together with the following neighborhood information for them: $\sigma(p)=p+1$ for $p=a, ..., b$ and $\sigma(p)=p-1$ for $p=b+1, ..., c$. It can be easily seen that the neighborhood information for all players at a pure Nash equilibrium can always be decomposed into disjoint segments. Importantly, given the neighborhood information in segment $C(a,b,c)$ and the beliefs of its players, the opinions they could have in any pure Nash equilibrium that contains this segment are uniquely defined using Lemma \ref{lem:equilibrium}. In particular, the opinions of the players within a segment $C(a,b,c)$ are computed as follows. First, we set $z_b = s_b+\frac{s_{b+1}-s_b}{3}$ and $z_{b+1} = s_b+\frac{2(s_{b+1}-s_b)}{3}$. Then, we set $z_p = \frac{s_p+z_{p+1}}{2}$ if $a\leq p <b$, and $z_p = \frac{s_p+z_{p-1}}{2}$ if $b<p\leq c$.

We remark that the opinion vector implied by a segment is not necessarily consistent to the given neighborhood structure. So, we call segment $C(a,b,c)$ {\em legit} if $a\not=2$, $c\not=n-1$ (so that it can be part of a decomposition) and the uniquely defined opinions are consistent to the neighborhood information of the segment, i.e., if $|z_{\sigma(p)} - s_p| \leq |z_{p'} - s_p|$ for any pair of players $p, p'$ (with $p\neq p'$) in $C(a,b,c)$. This process appears in Algorithm~\ref{alg:Segment}.

A decomposition of neighborhood information for all players will consist of consecutive segments $C(a_1,b_1,c_1)$, $C(a_2,b_2,c_2)$, ..., $C(a_t,b_t,c_t)$ so that $a_1=1$, $c_t=n$, $a_\ell=c_{\ell-1}+1$ for $\ell=2, ..., t$. Such a decomposition will yield a pure Nash equilibrium if it consists of legit segments and, furthermore, the uniquely defined opinions of players in consecutive segments are consistent to the neighborhood information.

In particular, consider the directed graph $G$ that has two special nodes designated as the source and the sink, and a node for each legit segment $C(a,b,c)$. Note that $G$ has $\bigO(n^3)$ nodes. The source node is connected to all segment nodes $C(1,b,c)$ while all segment nodes $C(a,b,n)$ are connected to the sink. An edge from segment node $C(a,b,c)$ to segment node $C(a',b',c')$ exists if $a'=c+1$ and the uniquely defined opinions of players in the two segments are consistent to the neighborhood information in both of them. This consistency test has to check
\begin{enumerate}
\item whether the leftmost opinion $z_{a'}$ in segment $C(a',b',c')$ is indeed further away from the belief $s_c$ of player $c$ than the opinion $z_{c-1}$ of the designated neighbor of $c$ in segment $C(a,b,c)$, i.e., $|z_{c-1}-s_c|\leq |z_{a'}-s_c|$, and whether
\item whether the rightmost opinion $z_c$ in segment $C(a,b,c)$ is further away from the belief $s_{a'}$ of player $a'$ than the opinion $z_{a'+1}$ of the designated neighbor of $a'$ in segment $C(a',b',c')$, i.e., $|z_{a'+1}-s_{a'}|\leq |z_c-s_{a'}|$.
\end{enumerate}
By the definition of segments and of its edges, $G$ is acyclic. This process appears in Algorithm~\ref{alg:ConstructGraph}.

Based on the discussion above, there is a bijection between pure Nash equilibria and source-sink paths in $G$. In addition, we can assign a weight to each node of $G$ that is equal to the total cost of the players in the corresponding segment, i.e.,
$$\weight({C(a,b,c)}) = \sum_{a\leq p \leq c}{|z_p-s_p|}.$$
Then, the total weight of a source-sink path ${\cal P}$ is equal to the social cost of the corresponding pure Nash equilibrium, i.e,
$$\SC(\zz,\sss) = \sum_{C(a,b,c) \in {\cal P}}{\weight({C(a,b,c)})}.$$

Hence, standard algorithms for computing shortest or longest paths in directed acyclic graphs can be used not only to detect whether a pure Nash equilibrium exists, but also to compute the equilibrium of best or worst social cost.

\begin{theorem}\label{thm:complexity}
Given a $1$-COF game, deciding whether a pure Nash equilibrium exists can be done in polynomial time. Furthermore, computing a pure Nash equilibrium of highest or lowest social cost can be done in polynomial time as well.
\end{theorem}

\begin{algorithm}[p]
    \SetAlgoLined
    \KwIn{belief vector $\sss = (s_1, ..., s_n)$, parameters $a$, $b$, and $c$ such that $a\leq b<c$}
    \KwOut{opinion vector $\zz_{a:c} = (z_a, ..., z_c)$, segment weight, legit indicator}
	$\text{legit} := 0$ \\
    \If{$a\neq2$ or $c\neq n-1$}{
    	$\text{legit} := 1$\\
	    $z_b := s_b + \frac{1}{3}(s_{b+1}-s_b)$ \\
    $z_{b+1} := s_b + \frac{2}{3}(s_{b+1}-s_b)$ \\
    \For{$p:=b-1$ downto $a$}{
    	$z_p := \frac{1}{2}(s_p + z_{p+1})$
    }
    \For{$p:=b+2$ to $c$}{
    	$z_p := \frac{1}{2}(s_p + z_{p-1})$
    }
    \For{$p:=a+1$ to $b$}{
    	\If{$|z_{p-1} - s_p| < |z_{p+1} - s_p|$}{
    		$\text{legit} := 0$
    	}
    }
    \For{$p:=b+1$ to $c-1$}{
    	\If{$|z_{p+1} - s_p| < |z_{p-1} - s_p|$}{
    		$\text{legit} := 0$
    	}
    }
    $\text{weight} := 0$ \\
	\For{$p:=a$ to $c$}{
		$\text{weight} := \text{weight} +|z_p - s_p|$
	}
     	}
    return [$\zz_{a:c}$, weight, legit]
\caption{\texttt{Segment}}
\label{alg:Segment}
\end{algorithm}

\begin{algorithm}[p]
    \SetAlgoLined
    \KwIn{belief vector $\sss = (s_1, ..., s_n)$}
    \KwOut{a node-weighted directed acyclic graph $G$}
	
	$V := \emptyset$ \\
    \For{$a:=1$ to $n-1$}{
    \For{$b:=a$ to $n-1$}{
    \For{$c:=b+1$ to $n$}{
        [$\zz_{a:c}$, weight, legit] := \texttt{Segment}($\sss,a,b,c$) \\
        \If{legit $=1$}{
    		$C.a = a$, $C.b = b$, $C.c = c$, $C.\zz_{a:c} = \zz_{a:c}$, $C.weight = weight$ \\
    		$V := V \union C$
    	}
    }
    }
    }
	$V := V \union \{\text{source}, \text{sink}\}$\\
	$E := \emptyset$\\
	\For{$C \in V$}{
		\uIf{$C.a = 1$}{
		   $E := E \union (\text{source},C)$
		}
		\ElseIf{$C.c = n$}{
		   $E := E \union (C,\text{sink})$
		}
	}
	\For{all segment pairs $(C,D)$ such that $D.a = C.c +1$}{
		\If{$|C.z_{c-1} - s_{C.c}|\leq |D.z_{a} - s_{C.c}|$ and $|D.z_{a+1} - s_{D.a}|\leq |C.z_{c} - s_{D.a}|$}{
			$E := E \union (C,D)$
		}
	}
    return $G = (V,E)$
\caption{\texttt{ConstructGraph}}
\label{alg:ConstructGraph}
\end{algorithm}

\begin{example}\label{ex:algo-example}
Consider a $1$-COF game with four players with belief vector
$\sss = (0, 9, 12, 21)$. According to the discussion above, there are $10$ segments of the form $C(a,b,c)$ with $1\leq a\leq b< c\leq 4$, but it can be shown that only $3$ of them are legit; these are $C(1,1,2)$, $C(3,3,4)$ (see Figure~\ref{fig:algo-example}a), and $C(1,2,4)$ (see Figure~\ref{fig:algo-example}b). For example, segment $C(1,1,4)$, in which $\sigma(1)= 2$, $\sigma(2) = 1$, $\sigma(3) = 2$, and $\sigma(4) = 3$, corresponds to the opinion vector $(3,6,9,15)$. This is not consistent to the neighborhood information $\sigma(2)=1$ in the segment, as the belief of player $2$ coincides with the opinion of player $3$, while the opinion of player $1$ is further away. The resulting directed acyclic graph $G$ (see Figure \ref{fig:algo-example}c) implies that there exist two pure Nash equilibria for this $1$-COF game, namely the opinion vectors $(3,6,15,18)$ and $(5,10,11,16)$.
\end{example}

\begin{figure}[ht!]
\centering
\begin{subfigure}{\textwidth}
\centering
\begin{tikzpicture}[xscale=1, yscale=1]
  \filldraw (0,0) -- (4,0);
  \filldraw (6,0) -- (10,0);
  \filldraw[dashed] (4,0) -- (6,0);

  \filldraw (0,-0.1) node[align=center, below] {\small $0$};
  \filldraw (4,-0.1) node[align=center, below] {\small $9$};
  \filldraw (6,-0.1) node[align=center, below] {\small $12$};
  \filldraw (10,-0.1) node[align=center, below] {\small $21$};

  \filldraw (1.33,-0.1) node[align=center, below] {\small $3$};
  \filldraw (2.66,-0.1) node[align=center, below] {\small $6$};
  \filldraw (7.33,-0.1) node[align=center, below] {\small $15$};
  \filldraw (8.66,-0.1) node[align=center, below] {\small $18$};

  \filldraw (1.33,0) circle(1pt);
  \filldraw (2.66,0) circle(1pt);
  \filldraw (7.33,0) circle(1pt);
  \filldraw (8.66,0) circle(1pt);

  \filldraw ([xshift=-2pt,yshift=-2pt]0,0) rectangle ++(4pt,4pt);
  \filldraw ([xshift=-2pt,yshift=-2pt]4,0) rectangle ++(4pt,4pt);
  \filldraw ([xshift=-2pt,yshift=-2pt]6,0) rectangle ++(4pt,4pt);
  \filldraw ([xshift=-2pt,yshift=-2pt]10,0) rectangle ++(4pt,4pt);

  \draw[-stealth,thick] (0,0) to[bend left=30] (1.33,0);
  \draw[-stealth,thick] (4,0) to[bend right=30] (2.66,0);
  \draw[-stealth,thick] (6,0) to[bend left=30] (7.33,0);
  \draw[-stealth,thick] (10,0) to[bend right=30] (8.66,0);

\end{tikzpicture}
\caption{ \ }
\end{subfigure}
\begin{subfigure}{\textwidth}
\centering
\begin{tikzpicture}[xscale=1, yscale=1]
  \filldraw (0,0) -- (10,0);

  \filldraw (0,-0.1) node[align=center, below] {\small $0$};
  \filldraw (4,-0.1) node[align=center, below] {\small $9$};
  \filldraw (6,-0.1) node[align=center, below] {\small $12$};
  \filldraw (10,-0.1) node[align=center, below] {\small $21$};

  \filldraw (2.33,-0.1) node[align=center, below] {\small $5$};
  \filldraw (4.66,-0.1) node[align=center, below] {\small $10$};
  \filldraw (5.33,-0.1) node[align=center, below] {\small $11$};
  \filldraw (7.66,-0.1) node[align=center, below] {\small $16$};

  \filldraw (2.33,0) circle(1pt);
  \filldraw (4.66,0) circle(1pt);
  \filldraw (5.33,0) circle(1pt);
  \filldraw (7.66,0) circle(1pt);

  \filldraw ([xshift=-2pt,yshift=-2pt]0,0) rectangle ++(4pt,4pt);
  \filldraw ([xshift=-2pt,yshift=-2pt]4,0) rectangle ++(4pt,4pt);
  \filldraw ([xshift=-2pt,yshift=-2pt]6,0) rectangle ++(4pt,4pt);
  \filldraw ([xshift=-2pt,yshift=-2pt]10,0) rectangle ++(4pt,4pt);

  \draw[-stealth,thick] (0,0) to[bend left=20] (2.33,0);
  \draw[-stealth,thick] (4,0) to[bend left=30] (4.66,0);
  \draw[-stealth,thick] (6,0) to[bend right=30] (5.33,0);
  \draw[-stealth,thick] (10,0) to[bend right=20] (7.66,0);

\end{tikzpicture}
\caption{ \ }
\end{subfigure}
\begin{subfigure}{\textwidth}
\centering
\begin{tikzpicture}[xscale=1, yscale=1]

	\node[fill=white,draw,thick,rounded corners=2pt,minimum height=0.8cm,minimum width=1.3cm,font=\small] at (0,0) (source){source};
	
	\node[fill=white,draw,thick,rounded corners=2pt,minimum height=0.8cm,minimum width=1.3cm,font=\small] at ([xshift=18em,yshift=0em]source.center) (sink){sink};
	
	\node[fill=white,draw,thick,rounded corners=2pt,minimum height=0.8cm,minimum width=1.3cm,font=\small] at ([xshift=6em,yshift=2em]source.center) (v1){$C(1,1,2)$};
	
	\node[fill=white,draw,thick,rounded corners=2pt,minimum height=0.8cm,minimum width=1.3cm,font=\small] at ([xshift=12em,yshift=2em]source.center) (v2){$C(3,3,4)$};
	
	\node[fill=white,draw,thick,rounded corners=2pt,minimum height=0.8cm,minimum width=1.3cm,font=\small] at ([xshift=9em,yshift=-2em]source.center) (v3){$C(1,2,4)$};
		
	\draw[-stealth, thick] (source) to (v1);
	\draw[-stealth, thick] (v1) to (v2);
	\draw[-stealth, thick] (v2) to (sink);
	
	\draw[-stealth, thick] (source) to (v3);
	\draw[-stealth, thick] (v3) to (sink);

\end{tikzpicture}
\caption{ \ }
\end{subfigure}
\caption{The $1$-COF game considered in Example~\ref{ex:algo-example}. (a) The legit segments $C(1,1,2)$ and $C(3,3,4)$ which imply the opinion vector $(3,6,15,18)$. (b) The legit segment $C(1,2,4)$ which implies the opinion vector $(5,10,11,16)$. (c) The directed acyclic graph $G$ which shows that there exist two pure Nash equilibria in the game.}
\label{fig:algo-example}
\end{figure}
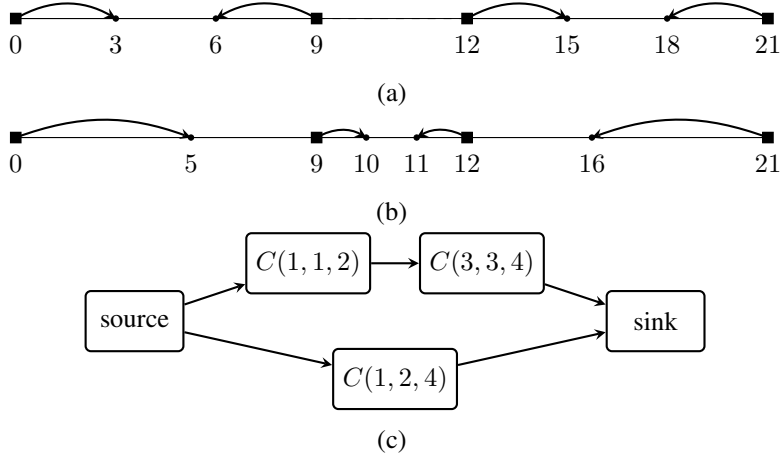

\section{Upper bounds on the price of anarchy}\label{sec:upper}
In this section we prove upper bounds on the price of anarchy of $k$-COF games. In our proof, we relate the social cost of any deterministic opinion vector, including optimal ones, to a quantity that depends only on the beliefs of the players and can be thought of as the cost of the truthful opinion vector (in which the opinion of every player is equal to her belief). In particular, we prove a lower bound on the optimal social cost (in Lemmas~\ref{lem:opt-lower}) and an upper bound on the social cost of any pure Nash equilibrium, both expressed in terms of this quantity. The bound on the price of anarchy then follows by these relations; see the proof of Theorem~\ref{thm:ppoa-upper}.

Consider an $n$-player $k$-COF game with belief vector $\sss=(s_1, ..., s_n)$. For player $i$, we denote by $\ell^*(i)$ and $r^*(i)$ the integers in $[n]$ such that $\ell^*(i)\leq i\leq r^*(i)$, $r^*(i)-\ell^*(i) = k$, and $|s_{r^*(i)}-s_{\ell^*(i)}|$ is minimized. The proof of the next lemma exploits linear programming and duality.

\begin{lemma}\label{lem:opt-lower}
Consider a $k$-COF game with belief vector $\sss = (s_1, ..., s_n)$ and let $\zz$ be any deterministic opinion vector. Then, $$\SC(\zz,\sss) \geq \frac{1}{2(k+1)}\sum_{i=1}^n{|s_{r^*(i)}-s_{\ell^*(i)}|}.$$
\end{lemma}

\begin{proof}
Consider any deterministic opinion vector $\zz$ and let $\pi$ be a permutation of the players so that $z_{\pi(j)}\leq z_{\pi(j+1)}$ for each $j\in [n-1]$. We refer to player $\pi(j)$ as the player with rank $j$.~\footnote{Note that we have proved monotonicity of opinions for pure Nash equilibria only (Lemma~\ref{lem:mon}) and it is not clear whether such a monotonicity property holds for opinion vectors of minimum social cost. In addition, the statement of Lemma~\ref{lem:opt-lower} refers to any opinion vector. This clearly includes non-monotonic ones, so we need to rank players in terms of opinions in the proof.} For each player $i$, we will identify an effective neighborhood $F_i(\zz,\sss)$ that consists of $k+1$ players with consecutive ranks and includes player $i$. Define $\tilde{\ell}(i)$ and $\tilde{r}(i)$ to be the players in $F_i(\zz,\sss)$ with the lowest and highest belief, respectively. In the extreme case where all players in $F_i(\zz,\sss)$ have the same belief, we let $\tilde{\ell}(i)$ and $\tilde{r}(i)$ be the players with the lowest and highest ranks, respectively. The effective neighborhood will be defined in such a way that it satisfies the properties $\cost_i(\zz,\sss)\geq z_{\tilde{r}(i)}-z_i$ and $\cost_i(\zz,\sss)\geq z_i-z_{\tilde{\ell}(i)}$.

Let $N_i(\zz,\sss)$ denote the neighborhood of player $i$, i.e., the set of players (not including $i$) with the $k$ closest opinions to the belief $s_i$ of player $i$. Let $J_i(\zz,\sss)$ be the smallest contiguous interval containing all opinions of players in $N_i(\zz,\sss)\cup \{i\}$ and let $D_i(\zz,\sss)$ be the set of players with opinions in $J_i(\zz,\sss)$. Clearly, $|D_i(\zz,\sss)|\geq k+1$. We define $F_i(\zz,\sss)$ to be a subset of $D_i(\zz,\sss)$ that consists of $k+1$ players with consecutive ranks including player $i$. See Figure~\ref{fig:example-opt-lower-k} for an illustrative example of all quantities defined above.

\begin{figure}[h!]
\centering
\begin{tikzpicture}[xscale=1, yscale=1]
  \filldraw[thick] (0,0) -- (12,0);

  \filldraw (0,-0.1) node[align=center, below] {\small $s_1$};
  \filldraw (1,-0.1) node[align=center, below] {\small $z_4$};
  \filldraw (2,-0.1) node[align=center, below] {\small $s_2$};
  \filldraw (3.1,-0.1) node[align=center, below] {\small $z_1$};
  \filldraw (3.8,-0.1) node[align=center, below] {\small $z_3$};
  \filldraw (5.4,-0.1) node[align=center, below] {\small $s_3$};
  \filldraw (6,-0.1) node[align=center, below] {\small $z_2$};
  \filldraw (7.7,-0.1) node[align=center, below] {\small $s_4$};
  \filldraw (6.8,-0.1) node[align=center, below] {\small $z_6$};
  \filldraw (9,-0.1) node[align=center, below] {\small $s_5$};
  \filldraw (10.5,-0.1) node[align=center, below] {\small $s_6$};
  \filldraw (12,-0.1) node[align=center, below] {\small $z_5$};

  \filldraw (3.1,0) circle(1pt);
  \filldraw (6,0) circle(1pt);
  \filldraw (3.8,0) circle(1pt);
  \filldraw (1,0) circle(1pt);
  \filldraw (6.8,0) circle(1pt);
  \filldraw (12,0) circle(1pt);

  \filldraw ([xshift=-2pt,yshift=-2pt]0,0) rectangle ++(4pt,4pt);
  \filldraw ([xshift=-2pt,yshift=-2pt]2,0) rectangle ++(4pt,4pt);
  \filldraw ([xshift=-2pt,yshift=-2pt]5.4,0) rectangle ++(4pt,4pt);
  \filldraw ([xshift=-2pt,yshift=-2pt]7.7,0) rectangle ++(4pt,4pt);
  \filldraw ([xshift=-2pt,yshift=-2pt]9,0) rectangle ++(4pt,4pt);
  \filldraw ([xshift=-2pt,yshift=-2pt]10.5,0) rectangle ++(4pt,4pt);

  \draw[-stealth,thick] (0,0) to[bend left=30] (3.1,0);
  \draw[-stealth,thick] (2,0) to[bend left=30] (6,0);
  \draw[-stealth,thick] (5.4,0) to[bend right=30] (3.8,0);
  \draw[-stealth,thick] (7.7,0) to[bend right=40] (1,0);
  \draw[-stealth,thick] (9,0) to[bend left=30] (12,0);
  \draw[-stealth,thick] (10.5,0) to[bend right=30] (6.8,0);

\end{tikzpicture}
\caption{An example of the quantities used in the proof of Lemma~\ref{lem:opt-lower}. Let $k=2$ and $i=4$. Then, the neighborhood of player $4$ is $N_4(\zz,\sss) = \{2,6\}$, the smallest contiguous interval containing the opinions of players in $N_4(\zz,\sss)\cup\{4\}$ is $J_4(\zz,\sss) = [z_4,z_6]$, the set of players with opinions in $J_4(\zz,\sss)$ is $D_4(\zz,\sss) = \{1,2,3,4,6\}$, the effective neighborhood is $F_4(\zz,\sss) = \{1,3,4\}$, and, hence, $\tilde{\ell}(4)=1$, and $\tilde{r}(4)=4$. }
\label{fig:example-opt-lower-k}
\end{figure}
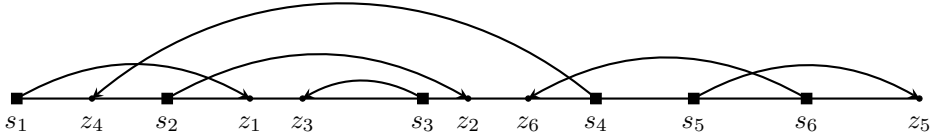

Let $\ell'(i)$ and $r'(i)$ be the players in $N_i(\zz,\sss)$ with the leftmost and rightmost opinion. In order to show that the definition of $F_i(\zz,\sss)$ satisfies the two desired properties, we distinguish between three different cases depending on the location of opinion $z_i$ among the players in $N_i(\zz,\sss)\cup \{i\}$.
\begin{itemize}
\item {\bf Case I:} Player $i$ has neither the leftmost nor the rightmost opinion in $N_i(\zz,\sss)\cup \{i\}$, i.e., $z_{\ell'(i)}<z_i<z_{r'(i)}$.~\footnote{Case I cannot appear when $k=1$.} In this case, $J_i(\zz,\sss)=[z_{\ell'(i)},z_{r'(i)}]$. Then, the definition of $N_i(\zz,\sss)$ implies that $\cost_i(\zz,\sss)\geq z_{r'(i)}-z_i$ and $\cost_i(\zz,\sss)\geq z_i-z_{\ell'(i)}$. Hence, $\cost_i(\zz,\sss)\geq |z_j-z_i|$ for every $z_j\in J_i(\zz,\sss)$ or, equivalently, $j\in D_i(\zz,\sss)$ and, subsequently, for each $j\in F_i(\zz,\sss)$. This implies the two desired properties $\cost_i(\zz,\sss)\geq z_{\tilde{r}(i)}-z_i$ and $\cost_i(\zz,\sss)\geq z_i - z_{\tilde{\ell}(i)}$.

\item {\bf Case II:} Player $i$ has the leftmost opinion in $N_i(\zz,\sss)\cup \{i\}$, i.e., $z_i\leq z_{\ell'(i)}$. Then, $J_i(\zz,\sss)=[z_i,z_{r'(i)}]$. Now, the definition of $N_i(\zz,\sss)$ implies that $\cost_i(\zz,\sss)\geq z_{r'(i)}-z_i$ and, hence, $\cost_i(\zz,\sss)\geq |z_j-z_i|$ for every $z_j\in J_i(\zz,\sss)$ or, equivalently, $j\in D_i(\zz,\sss)$ and, subsequently, for each $j\in F_i(\zz,\sss)$. Again, this implies the two desired properties.

\item {\bf Case III:} Player $i$ has the rightmost opinion in $N_i(\zz,\sss)\cup \{i\}$, i.e., $z_i\geq z_{r'(i)}$. Then, $J_i(\zz,\sss)=[z_{\ell'(i)},z_i]$. Now, the definition of $N_i(\zz,\sss)$ implies that $\cost_i(\zz,\sss)\geq z_i-z_{\ell'(i)}$ and, hence, $\cost_i(\zz,\sss)\geq |z_j-z_i|$ for every $z_j\in J_i(\zz,\sss)$ or, equivalently, $j\in D_i(\zz,\sss)$ and, subsequently, for every $j\in F_i(\zz,\sss)$. Again, the two desired properties follow.
\end{itemize}

By setting the variable $t_i$ equal to $\cost_i(\zz,\sss)$ for $i\in [n]$, the discussion above and the fact that $\cost_i(\zz,\sss)\geq |s_i-z_i|$ imply that the opinion vector $\zz$ together with $\ttt=(t_1, \dots, t_n)$ is a feasible solution to the following linear program:
\begin{align*}
\mbox{minimize} &\quad \sum_{i\in [n]}{t_i}\\
\mbox{subject to} &\quad t_i+z_i\geq s_i, \forall i\in [n]\\
&\quad t_i-z_i\geq -s_i, \forall i\in [n]\\
&\quad t_i+z_i-z_{\tilde{r}(i)} \geq 0, \forall i \in [n] \mbox{ such that } \tilde{r}(i)\not=i\\
&\quad t_i+z_{\tilde{\ell}(i)}-z_i \geq 0, \forall i \in [n] \mbox{ such that } \tilde{\ell}(i)\not=i\\
&\quad t_i,z_i\geq 0, \forall i\in [n]
\end{align*}

Using the dual variables $\alpha_i$, $\beta_i$, $\gamma_i$, and $\delta_i$ associated with the four constraints of the above LP, we obtain its dual LP:
\begin{align*}
\mbox{maximize} &\quad \sum_{i\in [n]}{s_i \alpha_i} - \sum_{i\in [n]}{s_i \beta_i}\\
\mbox{subject to} &\quad \alpha_i+\beta_i+\gamma_i \cdot\one{\tilde{r}(i)\not=i} +\delta_i \cdot\one{\tilde{\ell}(i)\not=i} \leq 1, \forall i\in [n]\\
&\quad \alpha_i-\beta_i+\gamma_i\cdot\one{\tilde{r}(i)\not=i}-\delta(i)\cdot\one{\tilde{\ell}(i)\not=i}-\sum_{j\not=i:\tilde{r}(j)=i}{\gamma_j} \\
&\quad +\sum_{j\not=i:\tilde{\ell}(j)=i}{\delta_j}\leq 0, \forall i\in [n]\\
&\quad \alpha_i, \beta_i, \gamma_i, \delta_i \geq 0
\end{align*}
The indicator $\one{X}$ is equal to $1$ when the condition $X$ is true, and $0$ otherwise.
We will show that the solution defined as $$\alpha_i=\frac{|\{j\in [n]:\tilde{r}(j)=i\}|}{2(k+1)},$$
$$\beta_i=\frac{|\{j\in [n]:\tilde{\ell}(j)=i\}|}{2(k+1)},$$
$$\gamma_i=\delta_i=\frac{1}{2(k+1)},$$
is a feasible dual solution. Indeed, to see why the first dual constraint is satisfied, first observe that player $i$ belongs to at most $2k+1$ different effective neighborhoods. Hence, player $i$ can have the lowest or highest belief among the players in the effective neighborhood of at most $2k+1$ players (implying that $\alpha_i+\beta_i\leq 1-\frac{1}{2(k+1)}$) when $\tilde{r}(i)=i$ or $\tilde{\ell}(i)=i$ and of at most $2k$ players  (implying that $\alpha_i+\beta_i\leq 1-\frac{1}{k+1}$) when $\tilde{r}(i)\not=i$ and $\tilde{\ell}(i)\not=i$. The first constraint follows.

It remains to show that the second constraint is satisfied as well (with equality). We do so by distinguishing between three cases:
\begin{itemize}
\item When $\tilde{r}(i)\not=i$ and $\tilde{\ell}(i)\not=i$, the dual solution guarantees that $\alpha_i=\sum_{j\not=i:\tilde{r}(j)=i}{\gamma_j}$ and the term $\alpha_i$ in the left-hand side of the second constraint cancels out with the sum of $\gamma$'s. Similarly, $\beta_i=\sum_{j\not=i:\tilde{\ell}(j)=i}{\delta_j}$ and the term $\beta_i$ cancels out with the sum of $\delta$'s. Also, the terms $\gamma_i$ and $\delta_i$ are both equal to $\frac{1}{2(k+1)}$ and cancel out as well.

\item When $\tilde{r}(i)=i$ (then, clearly, $\tilde{\ell}(i)\not=i$), we have that $\alpha_i=\delta_i\cdot \one{\tilde{\ell}(i)\not=i} +\sum_{j\not=i:\tilde{r}(j)=i}{\gamma_j}$ (cancelling out the first, fourth and fifth terms) and $\beta_i=\sum_{j\not=i:\tilde{\ell}(j)=i}{\delta_j}$ (cancelling out the second and sixth terms), and the second constraint is satisfied with equality as the third term is zero.

\item Finally, when $\tilde{\ell}(i)=i$ (now, it is $\tilde{r}(i)\not=i$), we have that $\alpha_i=\sum_{j\not=i:\tilde{r}(j)=i}{\gamma_j}$ (cancelling out the first and fifth terms) and $\beta_i=\gamma_i\cdot\one{\tilde{r}(i)\not=i}+\sum_{j\not=i:\tilde{\ell}(j)=i}{\delta_j}$ (cancelling out the second, third and sixth terms), and the second constraint is satisfied with equality as the fourth term is zero.
\end{itemize}

So, the social cost of the solution $\zz$ is lower-bounded by the objective value of the primal LP which, by duality, is lower-bounded by the objective value of the dual LP. Hence
\begin{align*}
\SC(\zz,\sss) &\geq \sum_{i\in [n]}{s_i \alpha_i} - \sum_{i\in [n]}{s_i \beta_i}\\
&= \frac{1}{2(k+1)} \left(\sum_{i\in [n]}{|\{j\in [n]:\tilde{r}(j)=i\}|s_i} - \sum_{i\in [n]}{|\{j\in [n]:\tilde{\ell}(j)=i\}|s_i}\right)\\
&= \frac{1}{2(k+1)} \sum_{i\in [n]}{(s_{\tilde{r}(i)}-s_{\tilde{\ell}(i)})}\\
&= \frac{1}{2(k+1)} \sum_{i\in [n]}{|s_{\tilde{r}(i)}-s_{\tilde{\ell}(i)}|}.
\end{align*}
The last equality follows since $s_{\tilde{r}(i)}\geq s_{\tilde{\ell}(i)}$, by the definition of $\tilde{r}(i)$ and $\tilde{\ell}(i)$.

Note that for each player $i$, there are at least $k+1$ beliefs of different players with values in $[s_{\tilde{\ell}(i)}, s_{\tilde{r}(i)}]$, including player $i$. By the definition of $\ell^*(i)$ and $r^*(i)$ for each player $i$, the above inequality yields
\begin{align*}
\SC(\zz,\sss) &\geq \frac{1}{2(k+1)} \sum_{i\in [n]}{|s_{r^*(i)}-s_{\ell^*(i)}|},
\end{align*}
as desired.
\end{proof}

We are now ready to prove our upper bound on the price of anarchy for $k$-COF games. In our proof, we exploit the mononicity of opinions in a pure Nash equilibrium and we associate the cost of each player in the equilibrium to the same quantity used in the statement of Lemma \ref{lem:opt-lower}.

\begin{theorem}\label{thm:ppoa-upper}
The price of anarchy of $k$-COF games over pure Nash equilibria is at most $4(k+1)$.
\end{theorem}

\begin{proof}
Consider a $k$-COF game with belief vector $\sss=(s_1, \dots, s_n)$, and let $\zz^*=(z_1^*, \dots, z_n^*)$ be any opinion vector that minimizes the social cost. By Lemma \ref{lem:opt-lower}, we have
\begin{align}\label{eq:opt-lower-k}
\SC(\zz^*,\sss) &\geq \frac{1}{2(k+1)}\sum_{i=1}^n{|s_{r^*(i)}-s_{\ell^*(i)}|}.
\end{align}

Now, consider any pure Nash equilibrium $\zz$ of the game. We will show that
\begin{align}\label{eq:nash-upper-k}
\SC(\zz,\sss) &\leq  2\sum_{i=1}^n{|s_{r^*(i)}-s_{\ell^*(i)}|},
\end{align}
and the theorem will then follow by inequalities (\ref{eq:opt-lower-k}) and (\ref{eq:nash-upper-k}).

The rest of this proof is, therefore, devoted to showing inequality (\ref{eq:nash-upper-k}). To this end, we will show that, for any player $i$, we have $\cost_i(\zz,\sss) \leq 2(s_{r^*(i)}-s_{\ell^*(i)})$. Then, inequality (\ref{eq:nash-upper-k}) will follow by summing over all players.

Consider an arbitrary player $i$ and, without loss of generality, let us assume that $z_i\geq s_i$ (the case $z_i\leq s_i$ is symmetric). Recall that $\ell(i)$ and $r(i)$ denote the players in $N_i(\zz,\sss)\cup\{i\}$ with the leftmost and rightmost point, respectively, in $I_i(\zz,\sss)$ and note that $r(i) - \ell(i) = k$. First, observe that if $z_{r(i)}=z_i$, the assumption $z_i\geq s_i$ implies that all players in $N_i(\zz,\sss)\cup\{i\}$ have opinions at $s_i$ (since, by Lemma \ref{lem:equilibrium}, $z_i$ is in the middle of interval $I_i(\zz,\sss)$ at equilibrium). In this case, $\cost_i(\zz,\sss)=0$ and the desired inequality holds trivially. So, in the following, we assume that $r(i)>i$ and $z_{r(i)}>z_i$, i.e., $z_{r(i)}$ is at the right of $z_i$ which in turn is at the right of (or coincides with) $s_i$.

Recall that, for player $i$, $\ell^*(i)$ and $r^*(i)$ denote the integers in $[n]$ such that $\ell^*(i)\leq i\leq r^*(i)$, $r^*(i)-\ell^*(i) = k$, and $|s_{r^*(i)}-s_{\ell^*(i)}|$ is minimized. Since $r(i) - \ell(i) = r^*(i) - \ell^*(i) = k$, we distinguish between two main cases depending on the relative order of $r(i)$ and $r^*(i)$.

\paragraph{Case I} $r(i)>r^*(i)$ and $\ell(i)>\ell^*(i)$. Since $z_{r(i)}$ is at the right of $s_i$ and $\ell^*(i)$ does not belong to the neighborhood of player $i$ (while player $r(i)$ does so by definition), $z_{\ell^*(i)}$ is at the left of $s_i$ and, furthermore, $z_{r(i)}-s_i\leq s_i-z_{\ell^*(i)}$ or, equivalently,
\begin{align}\label{eq:z-r-i}
z_{r(i)} &\leq 2s_i-z_{\ell^*(i)}.
\end{align}
This yields
\begin{align}\label{eq:cost-1}
\cost_i(\zz,\sss) &= z_{r(i)}-z_i \leq 2s_i-z_{\ell^*(i)}-z_i.
\end{align}
These inequalities will be useful in several places of the proof for this case below.

If $z_{\ell^*(i)}\geq s_{\ell^*(i)}$ then, since $r^*(i)\geq i$ and $z_i\geq s_i$, inequality (\ref{eq:cost-1}) becomes $\cost_i(\zz,\sss)\leq s_i-s_{\ell^*(i)} \leq s_{r^*(i)}-s_{\ell^*(i)}$ and the desired inequality follows. So, in the following, we assume that $z_{\ell^*(i)}<s_{\ell^*(i)}$ i.e., $z_{\ell^*(i)}$ is (strictly) at the left of $s_{\ell^*(i)}$. Hence, $\ell^*(i)$ has her leftmost neighbor with $z_{\ell(\ell^*(i))}<z_{\ell^*(i)}$ and, by Lemma \ref{lem:equilibrium},
\begin{align}\label{eq:z-l-star}
z_{\ell^*(i)} &= \frac{z_{\ell(\ell^*(i))}+\max\{s_{\ell^*(i)},z_{r(\ell^*(i))}\}}{2}.
\end{align}

Since $r^*(i)-\ell^*(i)=k$ and $\ell(\ell^*(i))<\ell^*(i)$,  we have $r^*(i) - \ell(\ell^*(i)) > k$,  and, therefore, $r^*(i)$ does not belong to the neighborhood of $\ell^*(i)$. Hence,
$s_{\ell^*(i)}-z_{\ell(\ell^*(i))}\leq z_{r^*(i)}-s_{\ell^*(i)}$ or, equivalently
\begin{align}\label{eq:z-l-l-star}
z_{\ell(\ell^*(i))} &\geq 2s_{\ell^*(i)}-z_{r^*(i)} \geq 2s_{\ell^*(i)}-2s_i+z_{\ell^*(i)},
\end{align}
where the second inequality follows by our case assumption $z_{r^*(i)}\leq z_{r(i)}$ and inequality (\ref{eq:z-r-i}).

We now further distinguish between two cases, depending on whether player $i$ belongs to the neighborhood of player $\ell^*(i)$ or not.

\paragraph{Case I.1} $i\in N_{\ell^*(i)}(\zz,\sss)$; see also Figure \ref{fig:example-Thm10}a for an example of this case. Then, we have $z_i\leq z_{r(\ell^*(i))}$ and, subsequently,
\begin{align}\label{eq:max}
\max\{s_{\ell^*(i)},z_{r(\ell^*(i))}\} &\geq z_{r(\ell^*(i))} \geq z_i.
\end{align}
Using inequalities (\ref{eq:z-l-l-star}) and (\ref{eq:max}), (\ref{eq:z-l-star}) yields
\begin{align*}
z_{\ell^*(i)} &\geq s_{\ell^*(i)}-s_i+\frac{z_{\ell^*(i)}}{2}+\frac{z_i}{2},
\end{align*}
which implies that $z_{\ell^*(i)} \geq  2s_{\ell^*(i)}-2s_i+z_i$. Now, inequality (\ref{eq:cost-1}) becomes
\begin{align*}
\cost_i(\zz,\sss) &\leq 4s_i-2s_{\ell^*(i)}-2z_i
\leq 2s_i-2s_{\ell^*(i)}
\leq 2(s_{r^*(i)}-s_{\ell^*(i)})
\end{align*}
as desired. The second inequality follows since $z_i\geq s_i$ and the last one follows since $r^*(i)\geq i$.

\paragraph{Case I.2} $i\not\in N_{\ell^*(i)}(\zz,\sss)$; see also Figure \ref{fig:example-Thm10}b for an example. Then, we have $s_{\ell^*(i)}-z_{\ell(\ell^*(i))} \leq z_i-s_{\ell^*(i)}$, which implies that $z_{\ell(\ell^*(i))} \geq 2s_{\ell^*(i)}-z_i$. Using this inequality together with the fact that $\max\{s_{\ell^*(i)},z_{r(\ell^*(i))}\}\geq s_{\ell^*(i)}$, (\ref{eq:z-l-star}) yields
\begin{align*}
z_{\ell^*(i)} &\geq \frac{3s_{\ell^*(i)}-z_i}{2}
\end{align*}
and inequality (\ref{eq:cost-1}) becomes
\begin{align*}
\cost_i(\zz,\sss) &\leq 2s_i-\frac{3}{2}s_{\ell^*(i)}-\frac{z_i}{2}
\leq \frac{3}{2}s_i-\frac{3}{2}s_{\ell^*(i)}
\leq 2(s_{r^*(i)}-s_{\ell^*(i)}),
\end{align*}
as desired. The second last inequality follows since $z_i\geq s_i$ and the last one follows since $r^*(i)\geq i$.

\paragraph{Case II} $r(i)\leq r^*(i)$ and $\ell(i)\leq \ell^*(i)$. Since $z_i$ is in the middle of the interval $I_i(\zz,\sss)$ and $z_{r(i)}$ is the rightmost opinion in $I_i(\zz,\sss)$, we have $$z_i = \frac{\min\{s_i,z_{\ell(i)}\} + z_{r(i)}}{2} \leq \frac{z_{\ell(i)}+z_{r(i)}}{2}\leq \frac{z_{\ell^*(i)}+z_{r^*(i)}}{2}.$$ Since $s_i\leq z_i$, the last inequality yields
\begin{align}\label{eq:l-star-II}
z_{\ell^*(i)} &\geq 2s_i-z_{r^*(i)}.
\end{align}
We also have
\begin{align}\label{eq:cost-2}
\cost_i(\zz,\sss) &= z_{r(i)}-z_i \leq z_{r^*(i)}-z_i.
\end{align}

If $z_{r^*(i)}\leq s_{r^*(i)}$ then, since $s_{\ell^*(i)} \leq s_i\leq z_i$, inequality (\ref{eq:cost-2}) yields $\cost_i(\zz,\sss)\leq s_{r^*(i)}-s_i\leq s_{r^*(i)}-s_{\ell^*(i)}$, which is even stronger than the desired inequality. So, in the following we assume that $z_{r^*(i)}>s_{r^*(i)}$, i.e., $z_{r^*(i)}$ is at the right of $s_{r^*(i)}$. Since $z_{r^*(i)}$ is in the middle of the interval $I_{r^*(i)}(\zz,\sss)$, we have that $r(r^*(i)) > r^*(i)$ and, therefore,
\begin{align}\label{eq:z-r-star-II}
z_{r^*(i)} &= \frac{\min\{s_{r^*(i)},z_{\ell(r^*(i))}\}+z_{r(r^*(i))}}{2}.
\end{align}
Moreover, since $r(r^*(i)) - \ell^*(i)> r^*(i) - \ell^*(i) = k$, player $\ell^*(i)$ does not belong to the neighborhood of player $r^*(i)$. Hence, $z_{r(r^*(i))}-s_{r^*(i)}\leq s_{r^*(i)}-z_{\ell^*(i)}$ which, together with inequality (\ref{eq:l-star-II}), yields that
\begin{align}\label{eq:z-r-r-star}
z_{r(r^*(i))} &\leq 2s_{r^*(i)}-z_{\ell^*(i)}
\leq 2s_{r^*(i)}-2s_i+z_{r^*(i)}.
\end{align}

We now further distinguish between two cases, depending on whether player $i$ belongs to the neighborhood of player $r^*(i)$ or not.

\paragraph{Case II.1} $i\in N_{r^*(i)}(\zz,\sss)$; see also Figure \ref{fig:example-Thm10}c for an example. Then, using the fact that $\min\{s_{r^*(i)},z_{\ell(r^*(i))}\} \leq z_{\ell(r^*(i))} \leq z_i$ and inequality (\ref{eq:z-r-r-star}), equation (\ref{eq:z-r-star-II}) becomes $$z_{r^*(i)} \leq  \frac{z_i+2s_{r^*(i)}-2s_i+z_{r^*(i)}}{2}$$
and, equivalently, $z_{r^*(i)} \leq z_i+2s_{r^*(i)}-2s_i$. Hence, inequality (\ref{eq:cost-2}) yields
\begin{align*}
\cost_i(\zz,\sss) &\leq 2s_{r^*(i)}-2s_i
\leq 2(s_{r^*(i)}-s_{\ell^*(i)}),
\end{align*}
as desired. The last inequality follows since $\ell^*(i)\leq i$.

\paragraph{Case II.2} $i\not\in N_{r^*(i)}(\zz,\sss)$; see Figure \ref{fig:example-Thm10}d for an example. Since $i$ does not belong to the neighborhood of player $r^*(i)$ but player $r(r^*(i))$ does, we have that $z_{r(r^*(i))}-s_{r^*(i)} \leq s_{r^*(i)}-z_i$ or, equivalently, $z_{r(r^*(i))} \leq 2s_{r^*(i)}-z_i$. Then, using also the fact that $\min\{s_{r^*(i)},z_{\ell(r^*(i))}\} \leq s_{r^*(i)}$, equation (\ref{eq:z-r-star-II}) becomes $$z_{r^*(i)}\leq \frac{3s_{r^*(i)}-z_i}{2}$$ and (\ref{eq:cost-2}) yields
\begin{align*}
\cost_i(\zz,\sss) &\leq \frac{3}{2}(s_{r^*(i)}-z_i)
\leq \frac{3}{2}(s_{r^*(i)}-s_{\ell^*(i)}),
\end{align*}
which is even stronger than the desired inequality. The last inequality follows since $z_i\geq s_i$ and $\ell^*(i)\leq i$.

So, we have shown that in the pure Nash equilibrium $\zz$ and for any player $i$, we have that $\cost_i(\zz,\sss) \leq 2(s_{r^*(i)}-s_{\ell^*(i)})$. By summing over all players, we obtain inequality (\ref{eq:nash-upper-k}) and the theorem follows.
\end{proof}

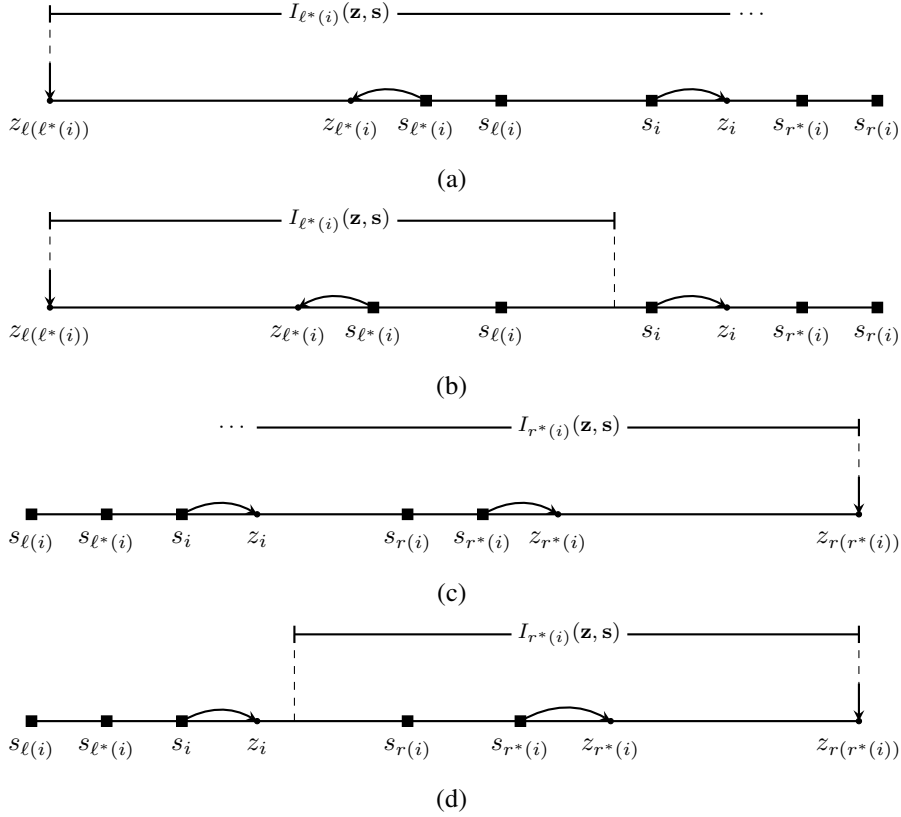
\begin{figure}[h!]
\begin{subfigure}{\textwidth}
\centering
\centering
\begin{tikzpicture}[xscale=1, yscale=1]
  \filldraw[thick] (0,0) -- (11,0);
  \filldraw (0,-0.1) node[align=center, below] {\small $z_{\ell(\ell^*(i))}$};
  \filldraw (4,-0.1) node[align=center, below] {\small $z_{\ell^*(i)}$};
  \filldraw (5,-0.1) node[align=center, below] {\small $s_{\ell^*(i)}$};
  \filldraw (6,-0.1) node[align=center, below] {\small $s_{\ell(i)}$};
  \filldraw (8,-0.1) node[align=center, below] {\small $s_i$};
  \filldraw (9,-0.1) node[align=center, below] {\small $z_i$};
  \filldraw (10,-0.1) node[align=center, below] {\small $s_{r^*(i)}$};
  \filldraw (11,-0.1) node[align=center, below] {\small $s_{r(i)}$};

  \filldraw (0,0) circle(1pt);
  \filldraw (4,0) circle(1pt);
  \filldraw (9,0) circle(1pt);

  \filldraw ([xshift=-2pt,yshift=-2pt]5,0) rectangle ++(4pt,4pt);
  \filldraw ([xshift=-2pt,yshift=-2pt]6,0) rectangle ++(4pt,4pt);
  \filldraw ([xshift=-2pt,yshift=-2pt]8,0) rectangle ++(4pt,4pt);
  \filldraw ([xshift=-2pt,yshift=-2pt]10,0) rectangle ++(4pt,4pt);
  \filldraw ([xshift=-2pt,yshift=-2pt]11,0) rectangle ++(4pt,4pt);

  \draw[-stealth,thick] (5,0) to[bend right=30] (4,0);
  \draw[-stealth,thick] (8,0) to[bend left=30] (9,0);

  \draw[-stealth,thick] (-0,0.5) to (0,0);

  \node at (0,0) (z){};

  \draw[thick] ([xshift=0em,yshift=3em]z.center) -- ([xshift=8em,yshift=3em]z.center);
  \draw[thick] ([xshift=0em,yshift=2.7em]z.center) -- ([xshift=0em,yshift=3.3em]z.center);
  \node[align=center] at ([xshift=10em,yshift=3em]z.center) {\scriptsize $I_{\ell^*(i)}(\zz,\sss)$};
  \draw[thick] ([xshift=12em,yshift=3em]z.center) -- ([xshift=23.5em,yshift=3em]z.center);
  \node[align=center] at ([xshift=24.3em,yshift=3em]z.center) {\scriptsize $\cdots$};

  \draw[dashed] ([xshift=0em,yshift=3em]z.center) -- (z.center);
\end{tikzpicture}
\caption{ \ }
\end{subfigure}
\begin{subfigure}{\textwidth}
\centering
\begin{tikzpicture}[xscale=1, yscale=1]
  \filldraw[thick] (0,0) -- (11,0);
  \filldraw (0,-0.1) node[align=center, below] {\small $z_{\ell(\ell^*(i))}$};
  \filldraw (3.3,-0.1) node[align=center, below] {\small $z_{\ell^*(i)}$};
  \filldraw (4.3,-0.1) node[align=center, below] {\small $s_{\ell^*(i)}$};
  \filldraw (6,-0.1) node[align=center, below] {\small $s_{\ell(i)}$};
  \filldraw (8,-0.1) node[align=center, below] {\small $s_i$};
  \filldraw (9,-0.1) node[align=center, below] {\small $z_i$};
  \filldraw (10,-0.1) node[align=center, below] {\small $s_{r^*(i)}$};
  \filldraw (11,-0.1) node[align=center, below] {\small $s_{r(i)}$};

  \filldraw (0,0) circle(1pt);
  \filldraw (3.3,0) circle(1pt);
  \filldraw (9,0) circle(1pt);

  \filldraw ([xshift=-2pt,yshift=-2pt]4.3,0) rectangle ++(4pt,4pt);
  \filldraw ([xshift=-2pt,yshift=-2pt]6,0) rectangle ++(4pt,4pt);
  \filldraw ([xshift=-2pt,yshift=-2pt]8,0) rectangle ++(4pt,4pt);
  \filldraw ([xshift=-2pt,yshift=-2pt]10,0) rectangle ++(4pt,4pt);
  \filldraw ([xshift=-2pt,yshift=-2pt]11,0) rectangle ++(4pt,4pt);

  \draw[-stealth,thick] (4.3,0) to[bend right=30] (3.3,0);
  \draw[-stealth,thick] (8,0) to[bend left=30] (9,0);

  \draw[-stealth,thick] (0,0.5) to (0,0);

  \node at (0,0) (z){};

  \draw[thick] ([xshift=0em,yshift=3em]z.center) -- ([xshift=8em,yshift=3em]z.center);
  \draw[thick] ([xshift=0em,yshift=2.7em]z.center) -- ([xshift=0em,yshift=3.3em]z.center);
  \node[align=center] at ([xshift=10em,yshift=3em]z.center) {\scriptsize $I_{\ell^*(i)}(\zz,\sss)$};
  \draw[thick] ([xshift=12em,yshift=3em]z.center) -- ([xshift=19.5em,yshift=3em]z.center);
  \draw[thick] ([xshift=19.5em,yshift=2.7em]z.center) -- ([xshift=19.5em,yshift=3.3em]z.center);

  \draw[dashed] ([xshift=0em,yshift=3em]z.center) -- (z.center);
  \draw[dashed] ([xshift=19.5em,yshift=3em]z.center) -- ([xshift=19.5em,yshift=0em]z.center);
\end{tikzpicture}
\caption{ \ }
\end{subfigure}
\begin{subfigure}{\textwidth}
\centering
\begin{tikzpicture}[xscale=1, yscale=1]
  \filldraw[thick] (0,0) -- (11,0);
  \filldraw (0,-0.1) node[align=center, below] {\small $s_{\ell(i)}$};
  \filldraw (1,-0.1) node[align=center, below] {\small $s_{\ell^*(i)}$};
  \filldraw (2,-0.1) node[align=center, below] {\small $s_i$};
  \filldraw (3,-0.1) node[align=center, below] {\small $z_i$};
  \filldraw (5,-0.1) node[align=center, below] {\small $s_{r(i)}$};
  \filldraw (6,-0.1) node[align=center, below] {\small $s_{r^*(i)}$};
  \filldraw (7,-0.1) node[align=center, below] {\small $z_{r^*(i)}$};
  \filldraw (11,-0.1) node[align=center, below] {\small $z_{r(r^*(i))}$};

  \filldraw (3,0) circle(1pt);
  \filldraw (7,0) circle(1pt);
  \filldraw (11,0) circle(1pt);

  \filldraw ([xshift=-2pt,yshift=-2pt]0,0) rectangle ++(4pt,4pt);
  \filldraw ([xshift=-2pt,yshift=-2pt]1,0) rectangle ++(4pt,4pt);
  \filldraw ([xshift=-2pt,yshift=-2pt]2,0) rectangle ++(4pt,4pt);
  \filldraw ([xshift=-2pt,yshift=-2pt]5,0) rectangle ++(4pt,4pt);
  \filldraw ([xshift=-2pt,yshift=-2pt]6,0) rectangle ++(4pt,4pt);

  \draw[-stealth,thick] (2,0) to[bend left=30] (3,0);
  \draw[-stealth,thick] (6,0) to[bend left=30] (7,0);

  \draw[-stealth,thick] (11,0.5) to (11,0);

  \node at (11,0) (z){};

  \draw[thick] ([xshift=0em,yshift=3em]z.center) -- ([xshift=-8em,yshift=3em]z.center);
  \draw[thick] ([xshift=0em,yshift=2.7em]z.center) -- ([xshift=0em,yshift=3.3em]z.center);
  \node[align=center] at ([xshift=-10em,yshift=3em]z.center) {\scriptsize $I_{r^*(i)}(\zz,\sss)$};
  \draw[thick] ([xshift=-12em,yshift=3em]z.center) -- ([xshift=-20.8em,yshift=3em]z.center);
  \node[align=center] at ([xshift=-21.6em,yshift=3em]z.center) {\scriptsize $\cdots$};

  \draw[dashed] ([xshift=0em,yshift=3em]z.center) -- (z.center);
\end{tikzpicture}
\caption{ \ }
\end{subfigure}
\begin{subfigure}{\textwidth}
\centering
\begin{tikzpicture}[xscale=1, yscale=1]
  \filldraw[thick](0,0) -- (11,0);

  \filldraw (0,-0.1) node[align=center, below] {\small $s_{\ell(i)}$};
  \filldraw (1,-0.1) node[align=center, below] {\small $s_{\ell^*(i)}$};
  \filldraw (2,-0.1) node[align=center, below] {\small $s_i$};
  \filldraw (3,-0.1) node[align=center, below] {\small $z_i$};
  \filldraw (5,-0.1) node[align=center, below] {\small $s_{r(i)}$};
  \filldraw (6.5,-0.1) node[align=center, below] {\small $s_{r^*(i)}$};
  \filldraw (7.7,-0.1) node[align=center, below] {\small $z_{r^*(i)}$};
  \filldraw (11,-0.1) node[align=center, below] {\small $z_{r(r^*(i))}$};

  \filldraw (3,0) circle(1pt);
  \filldraw (7.7,0) circle(1pt);
  \filldraw (11,0) circle(1pt);

  \filldraw ([xshift=-2pt,yshift=-2pt]0,0) rectangle ++(4pt,4pt);
  \filldraw ([xshift=-2pt,yshift=-2pt]1,0) rectangle ++(4pt,4pt);
  \filldraw ([xshift=-2pt,yshift=-2pt]2,0) rectangle ++(4pt,4pt);
  \filldraw ([xshift=-2pt,yshift=-2pt]5,0) rectangle ++(4pt,4pt);
  \filldraw ([xshift=-2pt,yshift=-2pt]6.5,0) rectangle ++(4pt,4pt);

  \draw[-stealth,thick] (2,0) to[bend left=30] (3,0);
  \draw[-stealth,thick] (6.5,0) to[bend left=30] (7.7,0);

  \draw[-stealth,thick] (11,0.5) to (11,0);

  \node at (11,0) (z){};

  \draw[thick] ([xshift=0em,yshift=3em]z.center) -- ([xshift=-8em,yshift=3em]z.center);
  \draw[thick] ([xshift=0em,yshift=2.7em]z.center) -- ([xshift=0em,yshift=3.3em]z.center);
  \node[align=center] at ([xshift=-10em,yshift=3em]z.center) {\scriptsize $I_{r^*(i)}(\zz,\sss)$};
  \draw[thick] ([xshift=-12em,yshift=3em]z.center) -- ([xshift=-19.5em,yshift=3em]z.center);
  \draw[thick] ([xshift=-19.5em,yshift=2.7em]z.center) -- ([xshift=-19.5em,yshift=3.3em]z.center);

  \draw[dashed] ([xshift=0em,yshift=3em]z.center) -- (z.center);
  \draw[dashed] ([xshift=-19.5em,yshift=3em]z.center) -- ([xshift=-19.5em,yshift=0em]z.center);
\end{tikzpicture}
\caption{ \ }
\end{subfigure}
\caption{Indicative examples of the different cases in the proof of Theorem \ref{thm:ppoa-upper}. Subfigures (a) and (b) concern Case I, as $r(i)>r^*(i)$ and $\ell(i)>\ell^*(i)$, while subfigures (c) and (d) fall under Case II, as $r(i)\leq r^*(i)$ and $\ell(i)\leq \ell^*(i)$.}
\label{fig:example-Thm10}
\end{figure}

\section{An improved bound on the price of anarchy for $1$-COF games}\label{sec:upper-1}

For the case of $1$-COF games we can prove an even stronger statement following a similar proof roadmap as in the previous section, but using simpler (and shorter) arguments. We denote by $\eta(i)$ the player (other than $i$) that minimizes the distance $|s_i-s_{\eta(i)}|$; note that $\eta(i) \in \{i-1,i+1\}$. The proof of the next lemma (which can be thought of as a stronger version of Lemma~\ref{lem:opt-lower} for 1-COF games) relies on a particular charging scheme that allows us to lower-bound the cost of each player in any deterministic opinion vector.

\begin{lemma}\label{lem:opt-lower-1}
Consider a $1$-COF game with belief vector $\sss = (s_1, \dots, s_n)$ and let $\zz$ be any deterministic opinion vector. Then, $$\SC(\zz,\sss) \geq \frac{1}{3}\sum_{i=1}^n{|s_i-s_{\eta(i)}|}.$$
\end{lemma}
\begin{proof}
We begin by classifying the players into groups and, subsequently, we show how the costs of different groups can be combined so that the lemma holds.  We call a player $i$ with $z_i\notin [s_{i-1}, s_{i+1}]$ a {\em kangaroo} player and associate the quantity $\excess_i$ with her. If $z_i \in [s_j, s_{j+1}]$ for some $j>i$, we say that the players in the set $C_i=\{i+1, ..., j\}$ are {\em covered} by player $i$ and define $\excess_i=z_i-s_j$. Otherwise, if $z_i \in [s_{j-1}, s_j]$ for some $j<i$, we say that the players in the set $C_i=\{j, ..., i-1\}$ are covered by player $i$ and define $\excess_i=s_j-z_i$.

Let ${\cal K}$ be the set of kangaroo players and ${\cal C}$ the set of players that are covered by a kangaroo; these need not be disjoint. We now partition the players not in ${\cal K} \cup {\cal C}$ into the set $L$ of {\em large} players such that, for any $i \in L$, it holds $\cost_i(\zz,\sss)\geq \frac{1}{3}(|s_i-s_{\eta(i)}|)$, and the set $S$ that contains the remaining players who we call {\em small}. See also Figure \ref{fig:example-opt-lower-1} for an example of these sets.

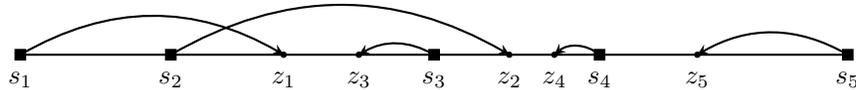
\begin{figure}[h!]
\centering
\begin{tikzpicture}[xscale=1, yscale=1]
  \filldraw[thick] (0,0) -- (11,0);

  \filldraw (0,-0.1) node[align=center, below] {\small $s_1$};
  \filldraw (2,-0.1) node[align=center, below] {\small $s_2$};
  \filldraw (3.5,-0.1) node[align=center, below] {\small $z_1$};
  \filldraw (4.5,-0.1) node[align=center, below] {\small $z_3$};
  \filldraw (5.5,-0.1) node[align=center, below] {\small $s_3$};
  \filldraw (6.5,-0.1) node[align=center, below] {\small $z_2$};
  \filldraw (7.1,-0.1) node[align=center, below] {\small $z_4$};
  \filldraw (7.7,-0.1) node[align=center, below] {\small $s_4$};
  \filldraw (9,-0.1) node[align=center, below] {\small $z_5$};
  \filldraw (11,-0.1) node[align=center, below] {\small $s_5$};

  \filldraw (3.5,0) circle(1pt);
  \filldraw (4.5,0) circle(1pt);
  \filldraw (6.5,0) circle(1pt);
  \filldraw (7.1,0) circle(1pt);
  \filldraw (9,0) circle(1pt);

  \filldraw ([xshift=-2pt,yshift=-2pt]0,0) rectangle ++(4pt,4pt);
  \filldraw ([xshift=-2pt,yshift=-2pt]2,0) rectangle ++(4pt,4pt);
  \filldraw ([xshift=-2pt,yshift=-2pt]5.5,0) rectangle ++(4pt,4pt);
  \filldraw ([xshift=-2pt,yshift=-2pt]7.7,0) rectangle ++(4pt,4pt);
  \filldraw ([xshift=-2pt,yshift=-2pt]11,0) rectangle ++(4pt,4pt);

  \draw[-stealth,thick] (0,0) to[bend left=30] (3.5,0);
  \draw[-stealth,thick] (2,0) to[bend left=30] (6.5,0);
  \draw[-stealth,thick] (5.5,0) to[bend right=30] (4.5,0);
  \draw[-stealth,thick] (7.7,0) to[bend right=40] (7.1,0);
  \draw[-stealth,thick] (11,0) to[bend right=30] (9,0);

\end{tikzpicture}
\caption{An example with kangaroos, covered, large, and small players. In particular, $1 \in {\cal K}$ as $z_1\notin [s_1,s_2]$, $2 \in {\cal K} \cap {\cal C}$ as she is covered by player $1$ and, in addition, $z_2 \notin [s_1, s_3]$. Similarly, $3 \in {\cal C}$ as she is covered by player $2$, while $4$ and $5$ are neither kangaroo nor covered. Since $\cost_4(\zz,\sss)< \frac{1}{3}(s_4-s_3)$,  it is $4 \in S$, while, since $\cost_5(\zz,\sss) \geq \frac{1}{3}(s_5-s_4)$, we have $5 \in L$.}
\label{fig:example-opt-lower-1}
\end{figure}

We proceed to prove five useful properties (Claims~\ref{claim:ineq-covered-kangaroos}--\ref{claim:ns2}); recall that $\sigma(i)$ denotes the single neighbor of player $i$.

\begin{claim}\label{claim:ineq-covered-kangaroos}
Let $i\in {\cal K}$. Then, $\cost_i(\zz,\sss)-\excess_i \geq \frac{1}{3}(|s_i-s_{\eta(i)}|+\sum_{j\in C_i}{|s_j-s_{\eta(j)}|})$.
\end{claim}

\begin{proof}
We assume that $z_i>s_i$ (the other case is symmetric). Let $\ell$ be the player with the rightmost belief that is covered by $i$. Then, $\excess_i=z_i-s_{\ell}$. We have
\begin{align*}
\cost_i(\zz,\sss)-\excess_i &= \max\{|s_i-z_i|,|z_i-z_{\sigma(i)}|\}-(z_i-s_\ell)\\
&\geq s_\ell-s_i=\sum_{j=i}^{\ell-1}(s_{j+1}-s_j)\\&\geq \frac{1}{3}(|s_i-s_{\eta(i)}|+\sum_{j\in C_i}{|s_j-s_{\eta(j)}|})\end{align*}
as desired.
\end{proof}

\begin{claim}\label{claim:ineq-small-kangaroos}
Let $i\in S$ such that $\sigma(i)\in {\cal K}$. Then, $\cost_i(\zz,\sss)+\excess_{\sigma(i)} \geq \frac{1}{3}|s_i-s_{\eta(i)}|$.
\end{claim}
\begin{proof}
We assume that $\sigma(i)>i$ (the other case is symmetric). If $z_{\sigma(i)}>s_{\sigma(i)}$, then
\begin{align*}
\cost_i(\zz,\sss) &= \max\{|s_i-z_i|,|z_i-z_{\sigma(i)}|\}\\
&\geq \frac{1}{2}(z_{\sigma(i)}-s_i) > \frac{1}{2}(s_{\sigma(i)}-s_i)\\
&\geq \frac{1}{3}|s_i-s_{\eta(i)}|,
\end{align*}
 which contradicts the fact that $i$ is a small player. Hence, $z_{\sigma(i)}\in [s_i,s_{\sigma(i)}]$, otherwise player $i$ would be covered. Let $j$ be the player with the leftmost belief that is covered by player $\sigma(i)$. Then, $\excess_{\sigma(i)}=s_j-z_{\sigma(i)}$. We have
 \begin{align*}
 \cost_i(\zz,\sss)+\excess_{\sigma(i)} &= \max\{|s_i-z_i|,|z_i-z_{\sigma(i)}|\}+s_j-z_{\sigma(i)}\\
  &\geq \frac{1}{2} (z_{\sigma(i)}-s_i)+\frac{1}{2}(s_j-z_{\sigma(i)}) = \frac{1}{2}(s_j-s_i)\\
  &\geq \frac{1}{3}|s_i-s_{\eta(i)}|
  \end{align*}
  as desired.
\end{proof}

\begin{claim}\label{claim:ineq-small-large}
Let $i\in S$ such that $\sigma(i)\in L$ or $\sigma(i)\in {\cal C}\setminus{\cal K}$. Then, $\cost_i(\zz,\sss)+\cost_{\sigma(i)}(\zz,\sss) \geq \frac{1}{3}(|s_i-s_{\eta(i)}|+|s_{\sigma(i)}-s_{\eta(\sigma(i))}|)$.
\end{claim}

\begin{proof}
We assume that $\sigma(i)>i$ (the other case is symmetric). If $z_{\sigma(i)}>s_{\sigma(i)}$, then
\begin{align*}
\cost_i(\zz,\sss) &= \max\{|s_i-z_i|,|z_i-z_{\sigma(i)}|\}\\
&\geq \frac{1}{2}(z_{\sigma(i)}-s_i) > \frac{1}{2}(s_{\sigma(i)}-s_i)\\
&\geq \frac{1}{3}|s_i-s_{\eta(i)}|,\end{align*}
 which contradicts the fact that $i$ is a small player. Hence, $z_{\sigma(i)}\in [s_i,s_{\sigma(i)}]$, otherwise player $i$ would be covered. Then, \begin{align*}
 \cost_i(\zz,\sss)+\cost_{\sigma(i)}(\zz,\sss)&=\max\{|s_i-z_i|,|z_i-z_{\sigma(i)}|\}+\max\{|s_{\sigma(i)}-z_{\sigma(i)}|,|z_{\sigma(i)}-z_{\sigma(\sigma(i))}|\}\\
 &\geq z_{\sigma(i)}-z_i+s_{\sigma(i)}-z_{\sigma(i)}=s_{\sigma(i)}-z_i.\end{align*}
  Since $i$ is small, we have $z_i<s_i+\frac{1}{3}(s_{\sigma(i)}-s_i)$ and we get
  $$
  \cost_i(\zz,\sss)+\cost_{\sigma(i)}(\zz,\sss) \geq \frac{2}{3}(s_{\sigma(i)}-s_i) \geq \frac{1}{3}|s_i-s_{\eta(i)}|+\frac{1}{3}|s_{\sigma(i)}-s_{\eta(\sigma(i))}|$$
  as desired. \end{proof}

Let $N(S)$ denote the set of players $j$ that are neighbors of players in $S$ (i.e., $j\in N(S)$ when $\sigma(i)=j$ for some player $i\in S$).

\begin{claim}\label{claim:ns1}
$N(S)$ does not contain small players.
\end{claim}
\begin{proof}
Assume otherwise that for some player $i\in S$, $\sigma(i)$ also belongs to $S$. Without loss of generality $\sigma(i)>i$. If $z_{\sigma(i)}\geq s_{\sigma(i)}$, then
$$\cost_i(\zz,\sss) \geq \frac{1}{2}|z_{\sigma(i)}-s_i|\geq \frac{1}{2}{|s_{\sigma(i)}-s_i|} \geq \frac{1}{3}|s_i-s_{\eta(i)}|$$
 contradicting the fact that $i\in S$. So, $z_{\sigma(i)}<s_{\sigma(i)}$. Also, $z_{\sigma(i)}\geq s_i$ (since neither $i$ is covered nor $\sigma(i)$ is kangaroo). Since $\sigma(i)$ is small, $s_{\sigma(i)}-z_{\sigma(i)}<\frac{1}{3}|s_{\sigma(i)}-s_{\eta(\sigma(i))}|\leq \frac{1}{3}(s_{\sigma(i)}-s_i)$, i.e., $z_{\sigma(i)}>\frac{2}{3}s_{\sigma(i)}+\frac{1}{3}s_i$. Hence,
 $$ \cost_i(\zz,\sss) \geq \frac{1}{2}(z_{\sigma(i)}-s_i) > \frac{1}{3}(s_{\sigma(i)}-s_i),$$ which contradicts $i\in S$.
\end{proof}

\begin{claim}\label{claim:ns2}
For every two players $i,i'\in S$, $\sigma(i)\not=\sigma(i')$.
\end{claim}
\begin{proof}
Assume otherwise and let $\sigma(i)=\sigma(i')=j$ with $i<i'$. If $z_j\not\in [s_i,s_{i'}]$, then the cost of either $i$ or $i'$ is at least $\frac{1}{2}(s_{i'}-s_i)$, contradicting the fact that both players are small. Hence, $z_j\in [s_i,s_{i'}]$. Notice that $s_j\in [s_i,s_{i'}]$ as well, otherwise either $i$ or $i'$ would be covered by $j$. Now the fact that $i$ and $i'$ are small implies that $$\cost_i(\zz,\sss)+\cost_{i'}(\zz,\sss)<\frac{1}{3}|s_i-s_{\eta(i)}|+\frac{1}{3}|s_{i'}-s_{\eta(i')}|\leq \frac{1}{3}(s_j-s_i)+\frac{1}{3}(s_{i'}-s_j)=\frac{1}{3}(s_{i'}-s_i).$$ On the other hand, $$\cost_i(\zz,\sss)+\cost_{i'}(\zz,\sss)\geq \frac{1}{2}(z_j-s_i)+\frac{1}{2}(s_{i'}-z_j)=\frac{1}{2}(s_{i'}-s_i),$$ a contradiction.
\end{proof}

We now consider the social cost of $\zz$ due to players of different groups and exploit the claims above so that we obtain the lemma. In particular, we have
\begin{align*}
\SC(\zz,\sss) &= \sum_{i=1}^n{\cost_i(\zz,\sss)}\\
&\geq \sum_{i\in S:\sigma(i)\in {\cal K}}{\left(\cost_i(\zz,\sss)+\excess_{\sigma(i)}\right)} \\
&\quad + \sum_{i\in S:\sigma(i)\in L\cup({\cal C}\setminus{\cal K})}{\left(\cost_i(\zz,\sss)+\cost_{\sigma(i)}(\zz,\sss)\right)}\\
&\quad + \sum_{i\in {\cal K}}{\left(\cost_i(\zz,\sss)-\excess_i\right)}
+ \sum_{i\in L\setminus N(S)}{\cost_i(\zz,\sss)}\\
&\geq \frac{1}{3}\sum_{i\in S:\sigma(i)\in {\cal K}}{|s_i-s_{\eta(i)}|} \\
&\quad + \frac{1}{3}\sum_{i\in S:\sigma(i)\in L\cup({\cal C}\setminus{\cal K})}{\left(|s_i-s_{\eta(i)}|+|s_{\sigma(i)}-s_{\eta(\sigma(i))}|\right)}\\
&\quad + \frac{1}{3}\sum_{i\in {\cal K}}{\left(|s_i-s_{\eta(i)}|+\sum_{j\in C_i}{|s_j-s_{\eta(j)}|}\right)}
+ \frac{1}{3}\sum_{i\in L\setminus N(S)}{|s_i-s_{\eta(i)}|}\\
&\geq \frac{1}{3}\sum_{i=1}^n{|s_i-s_{\eta(i)}|},
\end{align*}
as desired. The first inequality follows by the classification of the players and due to Claims \ref{claim:ns1} and \ref{claim:ns2}. The second one follows by Claims \ref{claim:ineq-small-kangaroos}, \ref{claim:ineq-small-large}, and \ref{claim:ineq-covered-kangaroos}, and by the definition of large players. The last one follows since the players enumerated in the first two sums at its left cover the whole set $S$ (by Claim \ref{claim:ns1}).
\end{proof}

We are ready to present our upper bound on the price of anarchy for $1$-COF games.

\begin{theorem}\label{thm:poa-upper-1}
The price of anarchy of $1$-COF games over pure Nash equilibria is at most~$3$.
\end{theorem}

\begin{proof}
Let us consider a $1$-COF game with $n$ players and belief vector $\sss$. Let $\zz^*$ be an optimal opinion vector and recall that $\eta(i)$ is the player that minimizes the distance $|s_i-s_{\eta(i)}|$.
By Lemma \ref{lem:opt-lower-1}, we have
\begin{align}\label{eq:opt-lower-1}
\SC(\zz^*,\sss) &\geq \frac{1}{3}\sum_{i=1}^n{|s_i-s_{\eta(i)}|}.
\end{align}

Now, consider any pure Nash equilibrium $\zz$ of the game. We will show that
\begin{align}\label{eq:nash-upper}
\SC(\zz,\sss) &\leq  \sum_{i=1}^n{|s_i-s_{\eta(i)}|}.
\end{align}
The theorem then follows by (\ref{eq:opt-lower-1}) and (\ref{eq:nash-upper}).

In particular, we will show that $\cost_i(\zz,\sss)\leq |s_i-s_{\eta(i)}|$ for each player $i$. Let us assume that $\eta(i)=i-1$; the case $\eta(i)=i+1$ is symmetric. Recall that $\sigma(i)$ is the neighbor of player $i$ in the pure Nash equilibrium $\zz$. We distinguish between four cases.
\begin{itemize}
\item {\bf Case I:} $\sigma(i)=i-1$. By Lemma \ref{lem:betw}, we have $s_{i-1}\leq z_i\leq s_i$. Then, clearly, $\cost_i(\zz,\sss) = |s_i-z_i| \leq |s_i-s_{i-1}|$ as desired.

\item {\bf Case II:} $\sigma(i)=i+1$ and $\sigma(i-1)=i$. By Lemmas \ref{lem:mon} and \ref{lem:betw}, we have $s_{i-1}\leq z_{i-1}\leq s_i \leq z_i$. Since player $i$ has player $i+1$ as neighbor, we have $|z_{i+1}-s_i|\leq |s_i-z_{i-1}|$. Hence, $\cost_i(\zz,\sss) =|z_i-s_i|\leq |z_{i+1}-s_i|\leq |s_i-z_{i-1}|\leq |s_i-s_{i-1}|$.

\item {\bf Case III:} $\sigma(i)=i+1$, $\sigma(i-1)=i-2$, and $\cost_i(\zz,\sss)\leq \cost_{i-1}(\zz,\sss)$. By the definition of $\sigma(\cdot)$ and Lemma \ref{lem:mon}, we have $z_{i-2}\leq z_{i-1}\leq s_{i-1}\leq s_i\leq z_i\leq z_{i+1}$. We have
\begin{align*}
\cost_i(\zz,\sss) &\leq 2\cost_{i-1}(\zz,\sss)-\cost_i(\zz,\sss)\\
&= |s_{i-1}-z_{i-2}|-|z_{i}-s_i|\\
&\leq |z_i-s_{i-1}|-|z_{i}-s_{i}|\\
&= |s_i-s_{i-1}|.
\end{align*}
The second inequality follows since player $i-2$ (instead of $i$) is the neighbor of player $i-1$.

\item {\bf Case IV:} $\sigma(i)=i+1$, $\sigma(i-1)=i-2$, and $\cost_i(\zz,\sss)>\cost_{i-1}(\zz,\sss)$.
\begin{align*}
\cost_i(\zz,\sss) &< 2\cost_i(\zz,\sss)-\cost_{i-1}(\zz,\sss)\\
&= |z_{i+1}-s_i|-|s_{i-1}-z_{i-1}|\\
&\leq |s_i-z_{i-1}|-|s_{i-1}-z_{i-1}|\\
&= |s_i-s_{i-1}|.
\end{align*}
The second inequality follows since player $i+1$ (instead of $i-1$) is the neighbor of player $i$.
\end{itemize}
This completes the proof.
\end{proof}


\section{Lower bounds on the price of anarchy}\label{sec:lower}
This section contains our lower bounds on the price of anarchy.~\footnote{We remark that our lower bounds on the price of stability in Section \ref{sec:existence} are also lower bounds on the price of anarchy. However, the lower bounds presented in this section are much stronger.} We begin by considering the simpler case of $1$-COF games, for which we present a tight lower bound of $3$ for pure Nash equilibria (Theorem \ref{thm:pure-lower-1}) and a lower bound of $6$ for mixed Nash equilibria (Theorem \ref{thm:mixed-lower-1}). We remark that, for $1$-COF games, this implies that mixed Nash equilibria are strictly worse than pure ones. Then, we study the general case of $k$-COF games and we show lower bounds for pure and mixed Nash equilibria (Theorems \ref{thm:pure-lower-k} and \ref{thm:mixed-lower-k}, respectively) that grow linearly with $k$.

\subsection{The case of $1$-COF games}
We now present our lower bounds for the case of $1$-COF games; both results rely on the same, and rather simple, instance.

\begin{figure}[h!]
\centering
\begin{subfigure}{\textwidth}
\centering
\begin{tikzpicture}[xscale=1, yscale=1]
  \filldraw (0,0) -- (4,0) -- (7,0) -- (11,0);

  \filldraw (0,-0.1) node[align=center, below] {\small $-10-\lambda$};
  \filldraw (0,0.1) node[align=center, above] {\small $[2]$};

  \filldraw (4,-0.1) node[align=center, below] {\small $-2-\lambda$};
  \filldraw (4,0.1) node[align=center, above] {\small $[1]$};

  \filldraw (7,-0.1) node[align=center, below] {\small $2+\lambda$};
  \filldraw (7,0.1) node[align=center, above] {\small $[1]$};

  \filldraw (11,-0.1) node[align=center, below] {\small $10+\lambda$};
  \filldraw (11,0.1) node[align=center, above] {\small $[2]$};

  \filldraw ([xshift=-2pt,yshift=-2pt]0,0) rectangle ++(4pt,4pt);
  \filldraw ([xshift=-2pt,yshift=-2pt]4,0) rectangle ++(4pt,4pt);
  \filldraw ([xshift=-2pt,yshift=-2pt]7,0) rectangle ++(4pt,4pt);
  \filldraw ([xshift=-2pt,yshift=-2pt]11,0) rectangle ++(4pt,4pt);
\end{tikzpicture}
\caption{ \ }
\end{subfigure}
\begin{subfigure}{\textwidth}
\centering
\begin{tikzpicture}[xscale=1, yscale=1]
  \filldraw (0,0) -- (11,0);

  \filldraw (0,-0.1) node[align=center, below] {\small $-10-\lambda$};
  \filldraw (4,-0.1) node[align=center, below] {\small $-2-\lambda$};
  \filldraw (7,-0.1) node[align=center, below] {\small $2+\lambda$};
  \filldraw (11,-0.1) node[align=center, below] {\small $10+\lambda$};

  \filldraw ([xshift=-2pt,yshift=-2pt]0,0) rectangle ++(4pt,4pt);
  \filldraw ([xshift=-2pt,yshift=-2pt]4,0) rectangle ++(4pt,4pt);
  \filldraw ([xshift=-2pt,yshift=-2pt]7,0) rectangle ++(4pt,4pt);
  \filldraw ([xshift=-2pt,yshift=-2pt]11,0) rectangle ++(4pt,4pt);

  \filldraw (2,0) circle(1pt);
  \filldraw (2,-0.1) node[align=center, below] {\small $-6-\lambda$};
  \filldraw (9,0) circle(1pt);
  \filldraw (9,-0.1) node[align=center, below] {\small $6+\lambda$};

  \node at (0,0) (zero){};
  \node at (11,0) (last){};

  \draw[-stealth,thick] (4,0) to[bend right=20] (2,0);
  \draw[-stealth,thick] (7,0) to[bend left=20] (9,0);
   \path[-stealth,thick] (zero) edge[out=45,in=135,looseness=10]  (zero);
  \path[-stealth,thick] (last) edge[out=135,in=45,looseness=10]  (last);
\end{tikzpicture}
\caption{ \ }
\end{subfigure}
\begin{subfigure}{\textwidth}
\centering
\begin{tikzpicture}[xscale=1, yscale=1]
  \filldraw (0,0) -- (11,0);

  \filldraw (0,-0.1) node[align=center, below] {\small $-10-\lambda$};

  \filldraw (4,-0.1) node[align=center, below] {\small $-2-\lambda$};

  \filldraw (7,-0.1) node[align=center, below] {\small $2+\lambda$};

  \filldraw (11,-0.1) node[align=center, below] {\small $10+\lambda$};

  \filldraw ([xshift=-2pt,yshift=-2pt]0,0) rectangle ++(4pt,4pt);
  \filldraw ([xshift=-2pt,yshift=-2pt]4,0) rectangle ++(4pt,4pt);
  \filldraw ([xshift=-2pt,yshift=-2pt]7,0) rectangle ++(4pt,4pt);
  \filldraw ([xshift=-2pt,yshift=-2pt]11,0) rectangle ++(4pt,4pt);


  \filldraw (5,0) circle(1pt);
  \filldraw (5.1,-0.05) node[align=center, below] {\small $\frac{-2-\lambda}{3}$};
  \filldraw (6,0) circle(1pt);
  \filldraw (5.9,-0.05) node[align=center, below] {\small $\frac{2+\lambda}{3}$};

  \node at (0,0) (zero){};
  \node at (11,0) (last){};

  \draw[-stealth,thick] (4,0) to[bend left=30] (5,0);
  \draw[-stealth,thick] (7,0) to[bend right=30] (6,0);
  \path[-stealth,thick] (zero) edge[out=45,in=135,looseness=10]  (zero);
  \path[-stealth,thick] (last) edge[out=135,in=45,looseness=10]  (last);
\end{tikzpicture}
\caption{ \ }
\end{subfigure}
\caption{(a) The $1$-COF game considered in the proofs of Theorems~\ref{thm:pure-lower-1} and \ref{thm:mixed-lower-1}. (b) The pure Nash equilibrium vector $\zz$ (see the proof of Theorem~\ref{thm:pure-lower-1}) with social cost $8$. (c) The opinion vector $\tilde{\zz}$ with social cost $\frac{8+4\lambda}{3}$.}
\label{fig:lower-bound}
\end{figure}
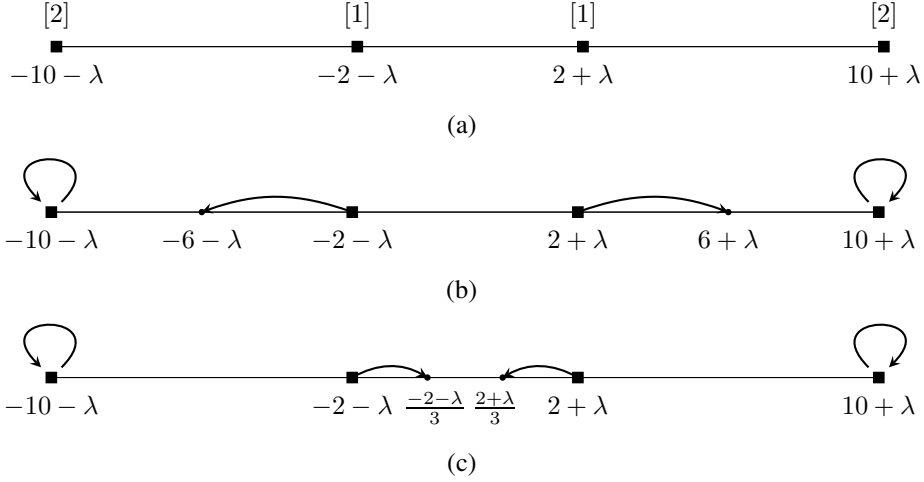

\begin{theorem}\label{thm:pure-lower-1}
The price of anarchy of $1$-COF games over pure Nash equilibria is at least $3$.
\end{theorem}

\begin{proof}
Let $\lambda \in (0,1)$ and consider a $1$-COF game with six players and belief vector
$\sss =(-10-\lambda,-10-\lambda,-2-\lambda,2+\lambda,10+\lambda,10+\lambda).$
This game is depicted in Figure~\ref{fig:lower-bound}a.
We can show that the opinion vector (see Figure~\ref{fig:lower-bound}b)
$$\zz = (-10-\lambda,-10-\lambda,-6-\lambda,6+\lambda,10+\lambda,10+\lambda)$$
is a pure Nash equilibrium with social cost $\SC(\zz, \sss) = 8$. The first two players suffer zero cost as they follow each other and their opinions coincide with their beliefs; the same holds also for the last two players. For the third player, it is $\sigma(3) \in \{1,2\}$ since $|z_1-s_3| = |z_2-s_3| = 8 < |z_4-s_3| = 8+2\lambda$ and $z_3$ is in the middle of the interval $[-10-\lambda,-2-\lambda]$; hence, $\cost_3(\zz,\sss) = 4$. Similarly, we have $\sigma(4) \in \{5,6\}$, $z_4$ lies in the middle of the interval $[2+\lambda, 10+\lambda]$ and $\cost_4(\zz,\sss) = 4$. Hence, $\zz$ is indeed a pure Nash equilibrium.

Now, consider the opinion vector (see Figure~\ref{fig:lower-bound}c)
$$\tilde{\zz} = \left(-10-\lambda,-10-\lambda,\frac{-2-\lambda}{3},\frac{2+\lambda}{3},10+\lambda,10+\lambda\right)$$
which yields a social cost of $\SC(\tilde{\zz}, \sss) = \frac{8+4\lambda}{3}$; here, again, the first and last two players have zero cost, but players $3$ and $4$ now each have cost $\frac{4+2\lambda}{3}$ since they follow each other. The optimal social cost is upper bounded by $\SC(\tilde{\zz})$ and, hence, the price of anarchy is at least
$$\frac{\SC(\zz, \sss)}{\SC(\tilde{\zz}, \sss)} = \frac{3}{1+\lambda/2},$$
and the theorem follows by setting $\lambda$ arbitrarily close to $0$.
\end{proof}

Our next theorem gives a lower bound on the price of anarchy over mixed Nash equilibria for $1$-COF games; we remark that this lower bound is greater than the upper bound of Theorem \ref{thm:poa-upper-1} for the price of anarchy over pure Nash equilibria.

\begin{theorem}\label{thm:mixed-lower-1}
The price of anarchy of $1$-COF games over mixed Nash equilibria is at least $6$.
\end{theorem}

\begin{proof}
Consider again the $1$-COF game depicted in Figure~\ref{fig:lower-bound}a with six players and belief vector $\sss =(-10-\lambda,-10-\lambda,-2-\lambda,2+\lambda,10+\lambda,10+\lambda)$, where $\lambda \in (0,1)$. To simplify the following discussion, we will refer to the first two players as the $L$ players, the third player as player $\ell$, the fourth player as player $r$, and the last two players as the $R$ players.

Let $\zz$ be a randomized opinion vector according to which $z_i=s_i$ for every $i \in L \cup R$, $z_\ell$ is chosen equiprobably from $\{-6-\lambda, -6+3\lambda\}$, and $z_r$ is chosen equiprobably from $\{6+\lambda, 6-3\lambda\}$. Observe that $\sigma(\ell) \in L$ whenever $z_r = 6+\lambda$, and $\sigma(\ell) = r$ whenever $z_r = 6-3\lambda$; each of these events occurs with probability $1/2$. Hence, we obtain
$$ \EE[\cost_{\ell}(\zz,\sss)] = \EE[\cost_r(\zz,\sss)] = \frac{1}{2}\left(\frac{4}{2}+\frac{4+4\lambda}{2}\right)+\frac{1}{2}\left(\frac{12-2\lambda}{2}+\frac{12-6\lambda}{2}\right) = 8-\lambda,$$ and, thus, $\EE[\SC(\zz,\sss)] = 16-2\lambda$. In the following, we will prove that $\zz$ is a mixed Nash equilibrium. First, observe that all players in sets $L$ and $R$ have no incentive to deviate since they follow each other and have zero cost. We will now argue about player $\ell$; due to symmetry, our findings will apply to player $r$ as well.

Consider a deterministic deviating opinion $y$ for player $\ell$. We will show that $\EE[\cost_{\ell}(\zz,\sss)]\leq \EE_{\zz_{-\ell}}[\cost_{\ell}(y,\zz_{-\ell}),\sss]$ for any $y$, which implies that player $\ell$ has no incentive to deviate from the randomized opinion $z_\ell$. Indeed, we have that
\begin{align*}
&\EE_{\zz_{-\ell}}[\cost_\ell((y,\zz_{-\ell}),\sss)] \\
&= \frac{1}{2}\max\{|-2-\lambda-y|,|y+10+\lambda|\}+\frac{1}{2}\max\{|-2-\lambda-y|,|6-3\lambda-y|\}\\
&\geq \frac{1}{2}(y+10+\lambda)+\frac{1}{2}(6-3\lambda-y)\\
&= 8-\lambda,
\end{align*}
where the inequality holds since $\max\{|a|,|b|\}\geq a$ for any $a$ and $b$. Hence, player $\ell$ has no incentive to deviate from her strategy in $\zz$, and neither has player $r$ due to symmetry. Therefore, $\zz$ is a mixed Nash equilibrium.

Now, consider the opinion vector
$$\tilde{\zz} = \left(-10-\lambda,-10-\lambda,\frac{-2-\lambda}{3},\frac{2+\lambda}{3},10+\lambda,10+\lambda \right)$$
which, as in Theorem~\ref{thm:pure-lower-1}, yields a social cost of $\SC(\tilde{\zz}, \sss) = \frac{8+4\lambda}{3}$. Hence, the optimal social cost is upper bounded by $\SC(\tilde{\zz}, \sss)$, and the price of anarchy over mixed equilibria is at least
$$\frac{\EE[\SC(\zz, \sss)]}{\SC(\tilde{\zz}, \sss)} = 3\frac{16-2\lambda}{8+4\lambda},$$
and the theorem follows by setting $\lambda$ arbitrarily close to $0$.
\end{proof}

\subsection{The general case of $k$-COF games with $k\geq 2$}

We will now present lower bounds on the price of anarchy for $k$-COF games, with $k\geq 2$. We start with the case of pure Nash equilibria and continue with the more general case of mixed equilibria. As in the case of $1$-COF games, a particular game will be used in order to derive the lower bounds both for pure and mixed Nash equilibria.

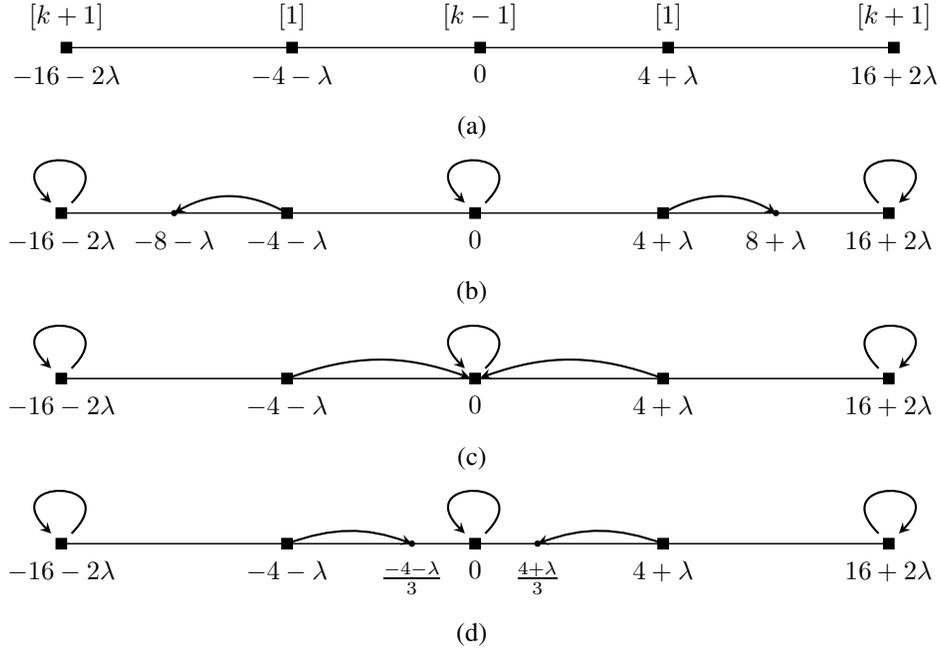
\begin{figure}[h!]
\centering
\begin{subfigure}{\textwidth}
\centering
\begin{tikzpicture}[xscale=1, yscale=1]
  \filldraw (0,0) -- (3,0) -- (5.5,0) -- (8,0) -- (11,0);

  \filldraw (0,-0.1) node[align=center, below] {\small $-16-2\lambda$};
  \filldraw (0,0.1) node[align=center, above] {\small $[k+1]$};

  \filldraw (3,-0.1) node[align=center, below] {\small $-4-\lambda$};
  \filldraw (3,0.1) node[align=center, above] {\small $[1]$};

  \filldraw (5.5,-0.1) node[align=center, below] {\small $0$};
  \filldraw (5.5,0.1) node[align=center, above] {\small $[k-1]$};

  \filldraw (8,-0.1) node[align=center, below] {\small $4+\lambda$};
  \filldraw (8,0.1) node[align=center, above] {\small $[1]$};

  \filldraw (11,-0.1) node[align=center, below] {\small $16+2\lambda$};
  \filldraw (11,0.1) node[align=center, above] {\small $[k+1]$};

  \filldraw ([xshift=-2pt,yshift=-2pt]0,0) rectangle ++(4pt,4pt);
  \filldraw ([xshift=-2pt,yshift=-2pt]3,0) rectangle ++(4pt,4pt);
  \filldraw ([xshift=-2pt,yshift=-2pt]5.5,0) rectangle ++(4pt,4pt);
  \filldraw ([xshift=-2pt,yshift=-2pt]8,0) rectangle ++(4pt,4pt);
  \filldraw ([xshift=-2pt,yshift=-2pt]11,0) rectangle ++(4pt,4pt);
\end{tikzpicture}
\caption{ \ }
\end{subfigure}
\begin{subfigure}{\textwidth}
\centering
\begin{tikzpicture}[xscale=1, yscale=1]
  \filldraw (0,0) -- (3,0) -- (5.5,0) -- (8,0) -- (11,0);

  \filldraw (0,-0.1) node[align=center, below] {\small $-16-2\lambda$};
  \filldraw (3,-0.1) node[align=center, below] {\small $-4-\lambda$};
  \filldraw (5.5,-0.1) node[align=center, below] {\small $0$};
  \filldraw (8,-0.1) node[align=center, below] {\small $4+\lambda$};
  \filldraw (11,-0.1) node[align=center, below] {\small $16+2\lambda$};

  \filldraw ([xshift=-2pt,yshift=-2pt]0,0) rectangle ++(4pt,4pt);
  \filldraw ([xshift=-2pt,yshift=-2pt]3,0) rectangle ++(4pt,4pt);
  \filldraw ([xshift=-2pt,yshift=-2pt]5.5,0) rectangle ++(4pt,4pt);
  \filldraw ([xshift=-2pt,yshift=-2pt]8,0) rectangle ++(4pt,4pt);
  \filldraw ([xshift=-2pt,yshift=-2pt]11,0) rectangle ++(4pt,4pt);

  \filldraw (1.5,0) circle(1pt);
  \filldraw (1.5,-0.1) node[align=center, below] {\small $-8-\lambda$};
  \filldraw (9.5,0) circle(1pt);
  \filldraw (9.5,-0.1) node[align=center, below] {\small $8+\lambda$};

  \node at (0,0) (zero){};
  \node at (5.5,0) (middle){};
  \node at (11,0) (last){};

  \draw[-stealth,thick] (3,0) to[bend right=30] (1.5,0);
  \draw[-stealth,thick] (8,0) to[bend left=30] (9.5,0);
  \path[-stealth,thick] (zero) edge[out=45,in=135,looseness=10]  (zero);
  \path[-stealth,thick] (middle) edge[out=45,in=135,looseness=10]  (middle);
  \path[-stealth,thick] (last) edge[out=135,in=45,looseness=10]  (last);
\end{tikzpicture}
\caption{ \ }
\end{subfigure}
\begin{subfigure}{\textwidth}
\centering
\begin{tikzpicture}[xscale=1, yscale=1]
  \filldraw (0,0) -- (3,0) -- (5.5,0) -- (8,0) -- (11,0);

  \filldraw (0,-0.1) node[align=center, below] {\small $-16-2\lambda$};
  \filldraw (3,-0.1) node[align=center, below] {\small $-4-\lambda$};
  \filldraw (5.5,-0.1) node[align=center, below] {\small $0$};
  \filldraw (8,-0.1) node[align=center, below] {\small $4+\lambda$};
  \filldraw (11,-0.1) node[align=center, below] {\small $16+2\lambda$};

  \filldraw ([xshift=-2pt,yshift=-2pt]0,0) rectangle ++(4pt,4pt);
  \filldraw ([xshift=-2pt,yshift=-2pt]3,0) rectangle ++(4pt,4pt);
  \filldraw ([xshift=-2pt,yshift=-2pt]5.5,0) rectangle ++(4pt,4pt);
  \filldraw ([xshift=-2pt,yshift=-2pt]8,0) rectangle ++(4pt,4pt);
  \filldraw ([xshift=-2pt,yshift=-2pt]11,0) rectangle ++(4pt,4pt);

  \node at (0,0) (zero){};
  \node at (5.5,0) (middle){};
  \node at (11,0) (last){};

  \draw[-stealth,thick] (3,0) to[bend left=20] (5.43,0.01);
  \draw[-stealth,thick] (8,0) to[bend right=20] (5.57,0.01);
  \path[-stealth,thick] (zero) edge[out=45,in=135,looseness=10]  (zero);
  \path[-stealth,thick] (middle) edge[out=45,in=135,looseness=10]  (middle);
  \path[-stealth,thick] (last) edge[out=135,in=45,looseness=10]  (last);
\end{tikzpicture}
\caption{ \ }
\end{subfigure}
\begin{subfigure}{\textwidth}
\centering
\begin{tikzpicture}[xscale=1, yscale=1]
  \filldraw (0,0) -- (3,0) -- (5.5,0) -- (8,0) -- (11,0);

  \filldraw (0,-0.1) node[align=center, below] {\small $-16-2\lambda$};
  \filldraw (3,-0.1) node[align=center, below] {\small $-4-\lambda$};
  \filldraw (5.5,-0.1) node[align=center, below] {\small $0$};
  \filldraw (8,-0.1) node[align=center, below] {\small $4+\lambda$};
  \filldraw (11,-0.1) node[align=center, below] {\small $16+2\lambda$};

  \filldraw ([xshift=-2pt,yshift=-2pt]0,0) rectangle ++(4pt,4pt);
  \filldraw ([xshift=-2pt,yshift=-2pt]3,0) rectangle ++(4pt,4pt);
  \filldraw ([xshift=-2pt,yshift=-2pt]5.5,0) rectangle ++(4pt,4pt);
  \filldraw ([xshift=-2pt,yshift=-2pt]8,0) rectangle ++(4pt,4pt);
  \filldraw ([xshift=-2pt,yshift=-2pt]11,0) rectangle ++(4pt,4pt);

  \filldraw (4.66,0) circle(1pt);
  \filldraw (4.66,-0.1) node[align=center, below] {\small $\frac{-4-\lambda}{3}$};
  \filldraw (6.33,0) circle(1pt);
  \filldraw (6.33,-0.1) node[align=center, below] {\small $\frac{4+\lambda}{3}$};

  \node at (0,0) (zero){};
  \node at (5.5,0) (middle){};
  \node at (11,0) (last){};

  \draw[-stealth,thick] (3,0) to[bend left=20] (4.66,0);
  \draw[-stealth,thick] (8,0) to[bend right=20] (6.33,0);
  \path[-stealth,thick] (zero) edge[out=45,in=135,looseness=10]  (zero);
  \path[-stealth,thick] (middle) edge[out=45,in=135,looseness=10]  (middle);
  \path[-stealth,thick] (last) edge[out=135,in=45,looseness=10]  (last);
\end{tikzpicture}
\caption{ \ }
\end{subfigure}
\caption{(a) The $k$-COF game considered in the proofs of Theorems~\ref{thm:pure-lower-k} and \ref{thm:mixed-lower-k}, for $k\geq 2$. (b) The pure Nash equilibrium opinion vector $\zz$ (see the proof of Theorem~\ref{thm:pure-lower-k}). (c) The optimal opinion vector $\tilde{\zz}$ for $k\geq 3$. (d) The optimal opinion vector $\tilde{\zz}$ for $k=2$. Observe that the optimal opinion vector changes at $k=2$ due to the neighborhood size.}
\label{fig:lower-bound-k}
\end{figure}

\begin{theorem}\label{thm:pure-lower-k}
The price of anarchy of $k$-COF games over pure Nash equilibria is at least $k+1$ for $k \geq 3$, and at least $18/5$ for $k=2$.
\end{theorem}

\begin{proof}
Let $\lambda \in (0,1)$ and consider a $k$-COF game with $3k+3$ players, for $k\geq 2$, that are partitioned into the following five sets. The first set $L$ consists of $k+1$ players with $s_i = -16-2\lambda$ for any $i\in L$, the second set consists of a single player $\ell$ with $s_{\ell} = -4-\lambda$, the third set $M$ has $k-1$ players with $s_i = 0$ for any $i\in M$, the fourth set is a single player $r$ with $s_r = 4+\lambda$, and the last set $R$ consists of $k+1$ players with $s_i= 16+2\lambda$ for any $i\in R$. This instance is depicted in Figure~\ref{fig:lower-bound-k}a.

Let $\zz$ be the following opinion vector: $z_i=-16-2\lambda$ for any $i \in R$, $z_\ell = -8-\lambda$, $z_i = 0$ for any $i \in M$, $z_r = 8+\lambda$, and $z_i = 16+2\lambda$ for any $i \in R$; see Figure~\ref{fig:lower-bound-k}b. It is not hard to verify that this opinion vector is a pure Nash equilibrium with social cost $\SC(\zz,\sss) = (8+\lambda)(k+1)$. First, observe that all players in sets $L$ and $R$ have zero cost, and, hence, have no incentive to deviate to another opinion.
Furthermore, no player $i \in M$ has an incentive to deviate either since $z_i$ lies in the middle of the interval $[-8-\lambda,8+\lambda]$ which is defined by the opinions of players $\ell$ and $r$ who, together with the remaining players of $M$, constitute the neighborhood $N_i(\zz,\sss)$ of player $i$. The cost experienced by such a player $i$ is $8+\lambda$.
Finally, the neighborhood $N_\ell(\zz,\sss)$ of player $\ell$ consists of all players in $M$ (who have opinions that are closest to $s_\ell$) and some player $i \in L$; note that player $r$ does not belong to $N_\ell(\zz,\sss)$ since $z_r - s_\ell = 12+2\lambda > 12-\lambda = s_\ell-z_i$ for all $i \in L$. Hence, player $\ell$ has no incentive to deviate to another opinion since $z_\ell$ lies in the middle of the interval $[-16-2\lambda,0]$ and she experiences cost equal to $8+\lambda$. Due to symmetry, player $r$ does not have incentive to deviate as well. Hence, $\zz$ is indeed a pure Nash equilibrium with $\SC(\zz,\sss) = (8+\lambda)(k+1)$.

We now present an opinion vector $\tilde{\zz}$ with social cost $\SC(\tilde{\zz},\sss) = 8+2\lambda$ for $k\geq 3$ and $\cost(\tilde{\zz},\sss) = \frac{5}{3}(4+\lambda)$ for $k=2$. In particular, for $k\geq 3$, $\tilde{\zz}$ is defined as follows: $\tilde{z}_i = -16-2\lambda$ for any $i \in L$, $\tilde{z}_\ell = \tilde{z}_i = \tilde{z}_r = 0$ for any $i \in M$, and $\tilde{z}_i = 16+2\lambda$ for any $i \in R$; see Figure~\ref{fig:lower-bound-k}c. Observe that all players in $L$, $M$, and $R$ have zero cost, while players $\ell$ and $r$ have cost equal to $4+\lambda$ each. For $k=2$, $\tilde{\zz}$ is defined as follows: $\tilde{z}_i = -16-2\lambda$ for any $i \in L$, $\tilde{z}_\ell = -\frac{1}{3}(4+\lambda)$, $\tilde{z}_i = 0$ for any $i \in M$, $\tilde{z}_r = \frac{1}{3}(4+\lambda)$, and $\tilde{z}_i = 16+2\lambda$ for any $i \in R$; see Figure~\ref{fig:lower-bound-k}d. Again, all players in $L$ and $R$ have zero cost. However, players $\ell$ and $r$ now each have cost $\frac{2}{3}(4+\lambda)$ and the unique player in $M$ has cost $\frac{1}{3}(4+\lambda)$.

Clearly, since $\SC(\tilde{\zz},\sss)$ is an upper bound on the optimal social cost, we conclude that the price of anarchy over pure Nash equilibria is at least $\frac{(8+\lambda)(k+1)}{8+2\lambda}$ for $k\geq 3$ and $\frac{9(8+\lambda)}{5(4+\lambda)}$ for $k=2$, and the theorem follows by setting $\lambda$ arbitrarily close to $0$.
\end{proof}

We now consider the case of mixed Nash equilibria; we remark that, in this case, our lower bounds for $k\geq 2$ are smaller than the corresponding upper bounds for pure Nash equilibria.

\begin{theorem}\label{thm:mixed-lower-k}
The price of anarchy of $k$-COF games over mixed Nash equilibria is at least $k+2$ for $k \geq 3$, and at least $24/5$ for $k=2$.
\end{theorem}

\begin{proof}
As in the proof of Theorem~\ref{thm:pure-lower-k}, let $\lambda \in (0,1)$ and consider the $k$-COF game depicted in Figure~\ref{fig:lower-bound-k}a with $3k+3$ players that form $5$ sets. Again, the first set $L$ consists of $k+1$ players where $s_i = -16-2\lambda$ for all $i\in L$, the second set consists of a single player $\ell$ with $s_{\ell} = -4-\lambda$, the third set $M$ has $k-1$ players with $s_i = 0$ for all $i\in M$, the fourth set is a single player $r$ with $s_r = 4+\lambda$, and the last set $R$ consists of $k+1$ players with $s_i = 16+2\lambda$ for all $i\in R$.

Consider the following (randomized) opinion vector $\zz$: $z_i = s_i$ for every $i \in L\cup M\cup R$, while $z_\ell$  is chosen uniformly at random among $\{-8-\lambda, -8+3\lambda\}$ and $z_r$ is chosen uniformly at random among $\{8-3\lambda,8+\lambda\}$. We will show that the opinion vector $\zz$ is a mixed Nash equilibrium with $\EE[\SC(\zz, \sss)] = 8k+16-\lambda$.

First, observe that the players in sets $L$ and the $R$ constitute local neighborhoods, that is, $N_i(\zz,\sss) = L\setminus\{i\}$ for any player $i \in L$, and $N_i(\zz,\sss)=R\setminus\{i\}$ for any player $i \in R$. Hence, all these players have zero cost and no incentive to deviate.

Next, let us focus on a player $i \in M$. Clearly, the neighborhood of player $i$ consists of the remaining $k-2$ players in $M$ as well as players $\ell$ and $r$. The expected cost of player $i$ in $\zz$ is $\EE[\cost_i(\zz,\sss)] = \frac{3}{4}(8+\lambda)+\frac{1}{4}(8-3\lambda) = 8$ since at least one of players $\ell$ and $r$ is at distance $8+\lambda$ with probability $3/4$ and both of them are at distance $8-3\lambda$ with probability $1/4$. Hence, these $k-1$ players contribute $8(k-1)$ to the expected social cost of $\zz$. We now argue that if player $i \in M$ deviates to a deterministic opinion $y$, her expected cost does not decrease. Clearly, if $y\geq 3\lambda$, then this trivially holds as the expected cost of $i$ is at least $y-z_\ell$ which is at least $y+8-3\lambda$; the case where $y\leq -3\lambda$ is symmetric. Hence, it suffices to consider the case where $|y|< 3\lambda$. The expected cost of $i$ when deviating to $y$ is
\begin{align*}
&\EE_{\zz_{-i}}[\cost_i((y,\zz_{-i}),\sss)] \\
&= \frac{1}{4}\max\{8+\lambda-y,y+8+\lambda\}+\frac{1}{4}\max\{8+\lambda-y,y+8-3\lambda\}\\
&\quad + \frac{1}{4}\max\{8-3\lambda-y,y+8+\lambda\} +\frac{1}{4}\max\{8-3\lambda-y,y+8-3\lambda\}\\
&\geq \frac{1}{4}(8+\lambda-y)+\frac{1}{4}(8+\lambda-y)+\frac{1}{4}(y+8+\lambda)+\frac{1}{4}(y+8-3\lambda)\\
&= 8,
\end{align*}
where the inequality holds since $\max\{a,b\}\geq a$ for any $a$ and $b$.

Now, let us examine player $r$; the case of player $\ell$ is symmetric. Observe that the $k-1$ players in $M$ always belong to the neighborhood $N_r(\zz,\sss)$ of player $r$ and it remains to argue about the identity of the last player in $N_r(\zz,\sss)$. Whenever $z_\ell = -8+3\lambda$, then $\ell \in N_r(\zz,\sss)$, otherwise, if $z_\ell = -8-\lambda$, one of the players in set $R$ belongs to $N_r(\zz,\sss)$. The expected cost of player $r$ is $\EE[\cost_r(\zz,\sss)] = \frac{1}{4}(8+\lambda) + \frac{1}{4}(8+5\lambda)+\frac{1}{4}(16-2\lambda)+\frac{1}{4}(16 - 6\lambda)= 12-\lambda/2$, and, hence, players $\ell$ and $r$ contribute $24-\lambda$ to the expected social cost of $\zz$. It remains to show that player $r$ cannot decrease her expected cost by deviating to another opinion $y$. The expected cost of player $r$ when deviating to $y$ is
\begin{align*}
\EE_{\zz_{-r}}[\cost_r((y,\zz_{-r}),\sss)] &= \frac{1}{2}\max\{|16+2\lambda-y|,|y|\}+\frac{1}{2}\max\{|y+8-3\lambda|,|4+\lambda-y|\}\\
&\geq \frac{1}{2}(16+2\lambda-y)+\frac{1}{2}(y+8-3\lambda)\\
&= 12-\lambda/2,
\end{align*}
where the inequality holds since $\max\{|a|,|b|\}\geq a$ for any $a$ and $b$. Hence, we conclude that $\zz$ is a mixed Nash equilibrium with expected social cost $\EE[\SC(\zz,\sss)] = 8k+16-\lambda$.

As in the proof of Theorem~\ref{thm:pure-lower-k}, there exists an opinion vector $\tilde{\zz}$ with social cost $\SC(\tilde{\zz},\sss) = 8+2\lambda$ for $k\geq 3$ and $\SC(\tilde{\zz},\sss) = \frac{5}{3}(4+\lambda)$ for $k=2$. Since $\SC(\tilde{\zz},\sss)$ is an upper bound on the optimal social cost, we have that the price of anarchy over mixed equilibria is at least $\frac{8k+16-\lambda}{8+2\lambda}$ for $k\geq 3$ and $\frac{3(32-\lambda)}{5(4+\lambda)}$ for $k=2$, and the theorem follows, again by setting $\lambda$ arbitrarily close to $0$.
\end{proof}

\section{Open problems and extensions}\label{sec:discussion}
We have introduced the class of compromising opinion formation games by enriching coevolutionary opinion games with a cost function that urges players to ``meet halfway''. Our findings indicate that the quality of their equilibria grows linearly with the neighborhood size $k$. Still, there exists a gap between our lower and upper bounds for $k\geq 2$ and closing this gap is a challenging technical task. Furthermore, we know that mixed equilibria are strictly worse for $1$-COF games but we have been unable to prove upper bounds on their price of anarchy. Is their price of anarchy still linear?

Another natural question is about the complexity of pure Nash equilibria in $k$-COF games for $k\geq 2$. We conjecture that there exists a polynomial time algorithm for computing them, but finding such an algorithm remains elusive.
Similarly, what is the complexity of computing an optimal opinion vector (even for $k=1$)?

Finally, our modeling assumption that the number of neighbors is the same for all players is rather restrictive. Extending our results to the general case of different neighborhood size per player deserves investigation. One possible such generalization is to combine our approach to the Hegselmann-Krause model, i.e., each player's neighborhood consists solely of those players with opinions sufficiently close to her belief.

\vskip 0.2in
\bibliographystyle{plainnat}
\bibliography{compromise}

\end{document}